\documentclass[letterpaper, 12pt]{article}
\usepackage{amssymb,amsmath,amsthm,mathtools, changes}
\usepackage{epsf}
\usepackage{times,bm,mathrsfs}
\usepackage{amsmath}
\usepackage{amssymb}
\usepackage{amsfonts}
\usepackage{mathrsfs}
\usepackage{bbm}
\usepackage{color}
\usepackage{graphicx}
\usepackage{amscd}
\usepackage{epsfig}
\usepackage{comment}
\usepackage{tikz, tikz-cd}
\usetikzlibrary{fit, positioning, arrows}

\definecolor{Red}{rgb}{1,0,0}
\definecolor{Blue}{rgb}{0,0,1}

\newtheorem{theorem}{Theorem}
\newtheorem{definition}{Definition}
\newtheorem{assumption}{Assumption}
\newtheorem{lemma}{Lemma}

\newtheorem{remark}{Remark}
\newtheorem{prop}{Proposition}

\def\beq{ \begin{equation} }
	\def\eeq{ \end{equation} }

\def\square{\vcenter{\vbox{\hrule height .4pt
			\hbox{\vrule width .4pt height 5pt \kern 5pt
				\vrule width .4pt} \hrule height .4pt}}}

\def\RR{\mathbb{R}}
\def\ZZ{\mathbb{Z}}

\def\var{\hbox{var}\,}

\def\P{{\mathbb P}}     
\def\E{{\mathbb E}}     
\def\<{{\langle}} 
\def\>{{\rangle}} 

\renewcommand{\var}{\mathrm{Var}}

\begin{document}
	\title{Statistically consistent and computationally efficient  inference of ancestral  DNA sequences in the TKF91 model under dense taxon sampling
		}
		\author{Wai-Tong (Louis) Fan\footnote{Department of Mathematics, 
				Indiana University.
				Work supported by NSF grant DMS--1149312 to SR and NSF grant DMS--1804492.}
			\and
			Sebastien Roch\footnote{Department of Mathematics, 
				UW--Madison.
				Work supported by NSF grants DMS-1149312 (CAREER), DMS-1614242, and CCF-1740707 (TRIPODS).}}
		
		\date{\today}
		\maketitle
		
		\begin{abstract}
		In evolutionary biology, the speciation history of living organisms is represented graphically by a phylogeny, that is, a rooted tree whose leaves correspond to current species and whose branchings indicate past speciation events. Phylogenetic analyses often rely on molecular sequences, such as DNA sequences, collected from the species of interest and
		it is common in this context to employ statistical approaches based on stochastic models of sequence evolution on a tree. For tractability, such models necessarily make simplifying assumptions about the evolutionary mechanisms involved. In particular, commonly omitted are \emph{insertions} and \emph{deletions} of nucleotides---also known as indels.
		
		Properly accounting for indels in statistical phylogenetic analyses remains a major challenge in computational evolutionary biology. 
		Here we consider the problem of reconstructing ancestral sequences on a known phylogeny in a model of sequence evolution incorporating nucleotide substitutions, insertions and deletions, specifically the classical TKF91 process. We focus on the case of dense phylogenies of bounded height, which we refer to as the taxon-rich setting, where statistical consistency is achievable. We give the first explicit reconstruction algorithm with provable guarantees under constant rates of mutation. Our algorithm succeeds when the phylogeny satisfies the ``big bang'' condition, a necessary and sufficient condition for statistical consistency in this setting.  
		\end{abstract}

		\thispagestyle{empty}		
		

		\newpage

\section{Introduction}		


\paragraph{Background}

In evolutionary biology, the speciation history of living organisms is represented graphically by a phylogeny, that is, a rooted tree whose leaves correspond to current species and branchings indicate past speciation events. Phylogenies are commonly estimated from molecular sequences, such as DNA sequences, collected from the species of interest. At a high level, the idea behind this inference is simple: the further apart in the Tree of Life are two species, the greater is the number of mutations to have accumulated in their genomes since their most recent common ancestor. 
In order to obtain accurate estimates in phylogenetic analyses, it is standard practice to employ statistical approaches based on stochastic models of sequence evolution on a tree. For tractability, such models necessarily make simplifying assumptions about the evolutionary mechanisms involved. In particular, commonly omitted are insertions and deletions of nucleotides---also known as indels. 
Properly accounting for indels in statistical phylogenetic analyses remains a major challenge in computational evolutionary biology.

Here we consider the related problem of reconstructing ancestral sequences on a known phylogeny in a model of sequence evolution incorporating nucleotide substitutions as well as indels. The model we consider, often referred to as the TKF91 process, was introduced in the seminal work of Thorne  et al.~\cite{thorne1991evolutionary} on the multiple sequence alignment problem (see also~\cite{thorne1992inching}). 
%
Much is known about ancestral sequence reconstruction (ASR) in substitution-only models. See, e.g.,~\cite{EvKePeSc:00,Mossel:01,Sly:09} and references therein,
as well as~\cite{liberles2007ancestral} for applications in biology. 
In the presence of indels, however, the only previous ASR result was obtained in~\cite{andoni2012global} for vanishingly small indel rates in a simplified version of the TKF91 process.   
The results in \cite{andoni2012global} 
concern what is known as ``solvability''; roughly, a sequence is inferred that exhibits a correlation with the true root sequence bounded away from 0 uniformly in the depth of the tree. 
The ASR problem in the presence of indels is also related to the trace reconstruction problem. See, e.g., the survey ~\cite{mitzenmacher2009survey} and references therein.

A desirable property of a reconstruction method is {\em statistical consistency}, which roughly says that the reconstruction is correct with probability tending to one as the amount of data increases.
It is known~\cite{EvKePeSc:00} that this is typically not information-theoretically achievable in the standard setting of the ASR problem. 
Here, however, we consider the {\em taxon-rich setting}, in which we have a sequence of  trees with {\it uniformly bounded heights} and {\it growing number of leaves}. 
Building on the work of Gascuel and Steel~\cite{GascuelSteel:10},
a necessary and sufficient condition for consistent ancestral reconstruction was derived in~\cite[Theorem 1]{MR3814241} in this context.

\paragraph{Our results}

In the current paper, our interest is in statistically consistent estimators  for the ASR problem under the TKF91 process in the taxon-rich setting, which differs from the ``solvability'' results in~\cite{andoni2012global}.
In fact, an ASR statistical consistency result in this context is already implied by the general results
of \cite{MR3814241}. However, the estimator in \cite{MR3814241} has drawbacks from a computational point of view. 
Indeed it relies on the computation of total variation distances between leaf distributions for different root states---and we are not aware of a tractable way to do these computations in TKF models.
The main contribution here is the design of an estimator which is not only {\it consistent} but also {\it constructive} and {\it computationally tractable}. 
We obtain this estimator 
by first estimating the {\it length} of the ancestral sequence, then estimating the sequence itself conditioned on the sequence length. 
The latter is achieved by deriving explicit formulas to invert the mapping from the root sequence to the distribution of the leaf sequences. (In statistical terms, we establish an identifiability result.) 

\paragraph{Further related results}
For tree reconstruction problems in the presence of indels, Daskalakis and Roch~\cite{DaskalakisRoch:13} devised the first polynomial-time phylogenetic tree reconstruction algorithm. In subsequent work, 
Ganesh and Zhang~\cite{Ganesh:2019:OSL:3313276.3316345} obtained a significantly improved sequence-length
requirement in a certain regime of parameters.
There is also a large computational biology
literature on methods for the co-estimation of phylogenetic trees and multiple sequence alignment,
typically without much theoretical guarantees. See, 
e.g., \cite{Warnow2013} and references therein.
Note that our approach to ASR does not require multiple sequence alignment.

\paragraph{Outline}
In Section~\ref{sec:defs}, we first recall the definition of the TKF91 process and our key assumptions on the sequence of trees. We then describe our new estimator and state the main results. In Section~\ref{S:ideas},  we give some intuition behind the definition of our estimator by giving a constructive proof of root-state identifiability for the TKF91 model. The proofs of the main results are given in Section~\ref{S:proofs}.  A summary discussion is in Section~\ref{S:discussion}. Some basic properties of the TKF91 process are derived in Section~\ref{sec:tkf-basic}.  A table of notation that are frequently used in this paper is provided in Section~\ref{S:Notation} for convenience of the reader.




\section{Definitions and main results}
\label{sec:defs}

Before stating our main results, we begin by describing the TKF91
model of Thorne, Kishino and Felsenstein~\cite{thorne1991evolutionary}, which
incorporates both substitutions and insertions/deletions (or {\bf indels} for short) in the evolution of a DNA sequence on a phylogeny. 
For simplicity, we follow Thorne et al.~and use the F81 model~\cite{Felsenstein:81} for the substitution component of the model,
although our results can be extended beyond this simple model. 
For ease of reference, a number of useful properties of the TKF91 model
are derived in Section~\ref{sec:tkf-basic}.

\subsection{TKF91 process}\label{SS:TKF91}

We first describe the Markovian dynamics on a single edge of the phylogeny. Conforming with the original definition of the model~\cite{thorne1991evolutionary},
we use an ``immortal link'' as a stand-in for the empty sequence.
\begin{definition}[TKF91 sequence evolution model on an edge]\label{Def_TKF91}	
The  {\bf TKF91 edge process} is a Markov process $\mathcal{I}=(\mathcal{I}_t)_{t\geq 0}$ on the space $\mathcal{S}$ of DNA sequences together with an {\bf immortal link}  ``$\bullet$", that is,
\begin{equation}\label{S}
\mathcal{S} := ``\bullet" \otimes \bigcup_{M\geq 0} \{A,T,C,G\}^M,
\end{equation}
where the notation above indicates that all sequences begin with the immortal link (and can otherwise be empty).
We also refer to the positions of a sequence (including nucleotides and the immortal link) as {\bf sites}. 
Let $(\nu,\,\lambda,\,\mu)\in (0,\infty)^3$
and $(\pi_A,\,\pi_T,\,\pi_C,\,\pi_G)\in [0,\infty)^4$ with $\pi_A +\pi_T + \pi_C + \pi_G = 1$ be given parameters. The continuous-time Markovian dynamic is described as follows: if the current state is the sequence $\vec{x}$, then the following events occur independently:
\begin{itemize}
	\item (Substitution)$\;$ Each nucleotide (but not the immortal link) is substituted independently at rate $\nu>0$. When a substitution occurs, the corresponding nucleotide is replaced by $A,T,C$ and $G$ with probabilities $\pi_A,\pi_T,\pi_C$ and $\pi_G$ respectively.
	
	\item (Deletion)$\;$ Each nucleotide (but not the immortal link) is removed independently at rate $\mu>0$.
	
	\item (Insertion) $\;$ Each site gives birth to a new nucleotide independently at rate $\lambda>0$. When a birth occurs, a nucleotide is  added immediately to the right of its parent site. The newborn site has nucleotide $A,T,C$ and $G$ with probabilities $\pi_A,\pi_T,\pi_C$ and $\pi_G$ respectively. 
	
\end{itemize}
The {\bf length} of a sequence $\vec{x}=(\bullet,x_1,x_2,\cdots,x_M)$ is defined as the number of nucleotides in $\vec{x}$ and is denoted by $|\vec{x}|=M$ (with the immortal link alone corresponding to $M=0$). When $M\geq 1$ we omit the immortal link for simplicity and write $\vec{x}=(x_1,x_2,\cdots,x_M)$.

\end{definition}
\noindent The TKF91 edge process is reversible~\cite{thorne1991evolutionary}. 
Suppose furthermore that 
$$
0 < \lambda < \mu,
$$ 
an assumption we make throughout. Then it has an {\bf stationary distribution} $\Pi$, given by
\begin{equation*}
\Pi(\vec{x})=
\left(1-\frac{\lambda}{\mu}\right) 
\left(\frac{\lambda}{\mu}\right)^M\prod_{i=1}^M\pi_{x_i} 
\end{equation*}
for each $\vec{x}=(x_1,x_2,\cdots,x_M)\in \{A,T,C,G\}^M$ where $M\geq 1$, and $\Pi(``\bullet") = \left(1-\frac{\lambda}{\mu}\right) $. In words, under $\Pi$, the sequence length is geometrically distributed and,
conditioned on the sequence length, all sites are independent with distribution $(\pi_{\sigma})_{\sigma\in \{A,T,C,G\}}$. 



The TKF91 edge process is a building block in the definition
of the sequence evolution model on a tree. 
Let $T= (V,E,\rho,\ell)$ be a {\bf phylogeny} (or, simply, tree), that is, a finite, edge-weighted, rooted tree, 
where $V$ is the set of vertices, $E$ is the set of edges oriented away from the root $\rho$, and $\ell:E \to (0,+\infty)$ is a positive edge-weighting function. We denote by $\partial T$ the leaf set of $T$. 
No assumption is made on the degree of the vertices. We think of $T$ as a continuous object, where each edge $e$ is a line segment of length $\ell_{e}$ and whose elements $\gamma \in e$ we refer to as {\bf points}. We let $\Gamma_T$ be the set of points of $T$. 
We consider the following stochastic process indexed by the points of 
$T$. 
\begin{definition}[TKF91 sequence evolution model on a tree]\label{Def_TKF91T}
The root is assigned a state $\vec{X}_{\rho}\in \mathcal{S}$, which is drawn from the stationary distribution $\Pi$ of the TKF91 edge process. The state is then evolved down the tree according to the following recursive process. Moving away from the root, along each edge $e = (u,v) \in E$,
conditionally on the state $\vec{X}_u$, we run  an independent TKF91 edge process as described in Definition \ref{Def_TKF91} started at $\vec{X}_u$ for an amount of time $\ell_{(u,v)}$. We denote by $\vec{X}_\gamma$ the resulting state at $\gamma \in e$. We call the process $(\vec{X}_\gamma)_{\gamma \in \Gamma_T}$ a {\bf TKF91 process on tree $T$}.  
\end{definition}
\noindent Our main interest is in the following statistical problem.

\subsection{Root reconstruction and the big bang condition}\label{SS:BB}

In the {\bf root reconstruction problem}, we seek to estimate the root state $\vec{X}_\rho$ based on the leaf states $\vec{X}_{\partial T}=\{\vec{X}_{v}:\,v\in \partial T\}$ of a TKF91 process on a known tree $T$. More formally, we look here for a consistent estimator, as defined next. 
Fix  mutation parameters $(\nu,\,\lambda,\,\mu)\in (0,\infty)^3$ with $\lambda < \mu$ and $(\pi_A,\,\pi_T,\,\pi_C,\,\pi_G)\in [0,\infty)^4$ and
let $\{T^k = (V^k, E^k, \rho^k, \ell^k)\}_{k \geq 1}$ be a sequence of trees with $|\partial T^k| \to +\infty$. Let $\mathcal{X}^k = (\vec{X}^k_\gamma)_{\gamma \in \Gamma_{T^k}}$ be a TKF91 process on $T^k$  defined on a probability space with probability measure $\P$.

\begin{definition}[Consistent root estimator]\label{Def:Consistent}
	A sequence of root estimators
	$$
	F_k:\mathcal{S}^{\partial T^k} \to \mathcal{S},
	$$
	is said to be {\em consistent} for the TKF91 process on $\{T^k\}_k$
	if
	$$
	\liminf_{k \to +\infty}
	\P
	\left[
	F_k\left(\vec{X}^k_{\partial T^k}\right) = \vec{X}^k_{\rho^k}
	\right] = 1.
	$$
\end{definition}
\noindent The mutation parameters and the sequence of trees are assumed to be known; that is, the estimators $F_k$ may depend on them. On the other hand, the leaf sequences $\vec{X}^k_{\partial T^k}$ are the only components of the process $\mathcal{X}^k$ that are actually observed.

As shown in~\cite{MR3814241}, in general a sequence of consistent estimators {\it may fail to exist}. 
Building on the work of Gascuel and Steel~\cite{GascuelSteel:10},
necessary and sufficient conditions for consistent root reconstruction are derived in~\cite[Theorem 1]{MR3814241} in the context of bounded-height, nested tree sequences with a growing number of leaves, which we refer to as the {\em taxon-rich setting}. These conditions have a combinatorial component (the big bang condition) and a stochastic component  (initial-state identifiability, to which we come back in Section~\ref{S:ideas}). To state the combinatorial condition formally, we need a few more definitions:

\begin{itemize}
\item 
{\it (Restriction)} Let $T = (V,E,\rho,\ell)$ be a tree.
For a subset of leaves $L \subset \partial T$,
the {\em restriction of $T$ to $L$} is the 
tree obtained from $T$ by keeping only those
points on a path between the root $\rho$ and 
a leaf $u \in L$.	

	\item {\it (Nested trees)} 
	We say that $\{T^k\}_{k}$ is a {\it nested sequence} if for all $k > 1$, $T^{k-1}$ is a restriction of $T^{k}$. 
	Without loss of generality, we assume that $|\partial T^k| = k$, so that $T^k$ is obtained by adding a leaf edge to $T^{k-1}$. (More general sequences can be obtained as subsequences.) In a slight abuse of notation, we denote by $\ell$ the edge-weight function for all $k$.
	For $\gamma\in \Gamma_T$, we denote by $\ell_{\gamma}$ the length of the unique path from the root $\rho$ to $\gamma$. We refer to $\ell_\gamma$ as the distance from $\gamma$ to the root. 
	
	\item {\it (Bounded height)} We further say that $\{T^k\}_{k}$ has {\it uniformly bounded height} if
	\begin{equation}\label{Uniformh}
	h^* := \sup_{k}h^{k} < +\infty,
	\end{equation}
	where $h^{k}:=\max\{\ell_x:\,x\in \partial T^k\}$ is the height of $T^k$. 
	
	
	\item {\it (Big bang)} For a tree $T = (V,E,\rho,\ell)$, let  
	$$
	T(s)=\{\gamma\in \Gamma_T:\;\ell_\gamma \leq s\}
	$$ 
	denote the tree obtained by truncating $T$ at distance $s$ from the root. We refer to $T(s)$ as a {\em truncation} of $T$. (See the left side of Figure \ref{Fig:Subtree} for an illustration.)
	We say that a sequence of trees $\{T^k\}_k$ satisfies the {\em big bang condition} if:
	for all $s\in(0,+\infty)$, we have $|\partial T^k(s)|\to +\infty$ as $k\to+\infty$.     (See Figure~\ref{Fig:BigB} for an illustration.)
    
   In words, the big bang condition ensures the existence of a large number of leaves that are ``almost conditionally independent'' given the root state, which is shown in~\cite{MR3814241} to be necessary for consistency. 
	
\end{itemize}

	\tikzset{
	big dot/.style={
		circle, inner sep=0pt, 
		minimum size=1.2mm, fill=black
	}
}

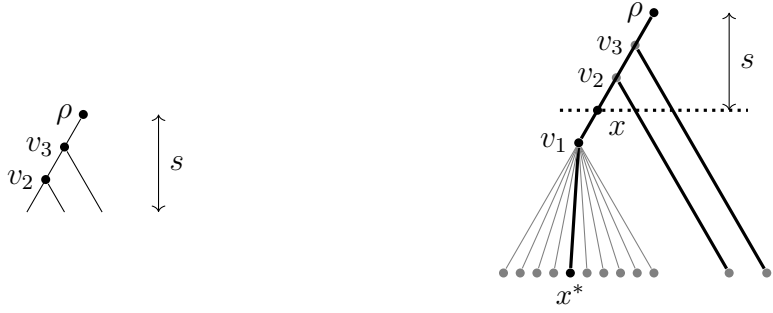
\begin{figure}
	\begin{minipage}{.20\textwidth}
		\quad
	\end{minipage}
	\begin{minipage}{.4\textwidth}
		\begin{tikzpicture}
		\node[big dot] (origin) at (0,0) {};
		\coordinate (X) at (-0.75,-1.299) {};
		\node[big dot] (B) at (-0.5,-0.866)   {};
		\coordinate (L3) at (-.25,-1.299)   {};
		\node[big dot] (C) at (-0.25,-0.433)   {};
		\coordinate (L4) at (.25,-1.299)   {};
		\draw[-] (origin) -- (X);
		\draw[-] (B) -- (L3);
		\draw[-] (C) -- (L4);
		\draw (origin) node[left]{$\rho$};
		\draw (B) node[left]{$v_2$};
		\draw (C) node[left]{$v_3$};
		\draw[<->] (1,0) -- node[anchor=west] {$s$} (1,-1.299);
		\end{tikzpicture}
	\end{minipage}
	\begin{minipage}{.4\textwidth}
		\begin{tikzpicture}
		\node[big dot] (origin) at (0,0) {};
		\node[big dot] (A) at (-1,-1.732) {};
		\node[big dot] (X) at (-0.75,-1.299) {};
		\draw (X) node[anchor=north west]{$x$};
		\node[big dot,color=black!50] (L1) at (-2,-3.464)   {};
		\node[big dot,color=black!50] (L2) at (0,-3.464)   {};
		\node[big dot,color=black!50] (B) at (-0.5,-0.866)   {};
		\node[big dot,color=black!50] (L3) at (1,-3.464)   {};
		\node[big dot,color=black!50] (C) at (-0.25,-0.433)   {};
		\node[big dot,color=black!50] (L4) at (1.5,-3.464)   {};
		\node[big dot,color=black!50] (K1) at (-0.22222,-3.464)   {};
		\node[big dot,color=black!50] (K2) at (-0.44444,-3.464)   {};		
		\node[big dot,color=black!50] (K3) at (-0.66666,-3.464)   {};
		\node[big dot,color=black!50] (K4) at (-0.88888,-3.464)   {};
		\node[big dot] (K5) at (-1.11111,-3.464)   {};
		\node[big dot,color=black!50] (K6) at (-1.33333,-3.464)   {};		
		\node[big dot,color=black!50] (K7) at (-1.55555,-3.464)   {};
		\node[big dot,color=black!50] (K8) at (-1.77777,-3.464)   {};
		
		\draw[-,color=black!50] (A) -- (L1);
		\draw[-,color=black!50] (A) -- (L2);
		\draw[-,very thick] (origin) -- (A);
		\draw[-,very thick] (B) -- (L3);
		\draw[-,very thick] (C) -- (L4);
		\draw[-,color=black!50] (A) -- (K1);
		\draw[-,color=black!50] (A) -- (K2);
		\draw[-,color=black!50] (A) -- (K3);
		\draw[-,color=black!50] (A) -- (K4);
		\draw[-,very thick] (A) -- (K5);
		\draw (K5) node[below]{\small{$x^*$}};
		\draw[-,color=black!50] (A) -- (K6);
		\draw[-,color=black!50] (A) -- (K7);
		\draw[-,color=black!50] (A) -- (K8);
		
		\draw (origin) node[left]{$\rho$};
		\draw (A) node[left]{$v_1$};
		\draw (B) node[left]{$v_2$};
		\draw (C) node[left]{$v_3$};
		
		\draw[<->] (1,0) -- node[anchor=west] {$s$} (1,-1.299);
		\draw[-,dotted,very thick] (-1.25,-1.299) -- (1.25,-1.299);    
		
		\end{tikzpicture}
	\end{minipage}
	\caption{On the left side, $T^k(s)$ is shown for the tree $T^k$ on the right, where the subtree $\tilde{T}^{k,s}$ is highlighted.}\label{Fig:Subtree}
\end{figure}
	\tikzset{
	big dot/.style={
		circle, inner sep=0pt, 
		minimum size=1.2mm, fill=black
	}
}
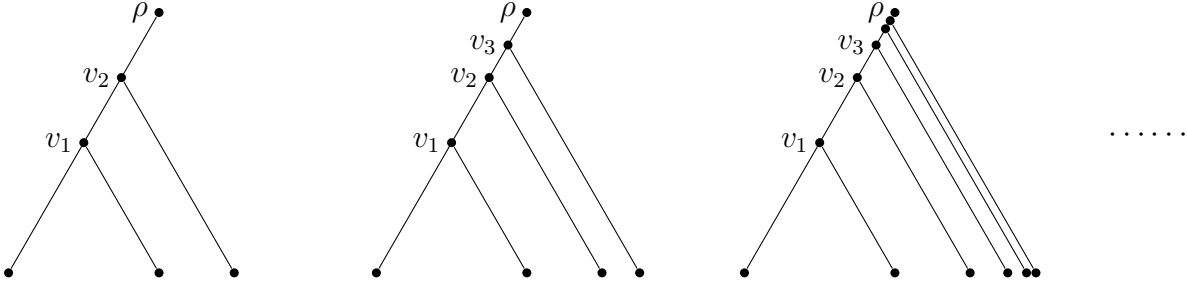
\begin{figure}
	\begin{minipage}{.29\textwidth}
		\begin{tikzpicture}
		\node[big dot] (origin) at (0,0) {};
		\node[big dot] (A) at (-1,-1.732) {};
		\node[big dot] (L1) at (-2,-3.464)   {};
		\node[big dot] (L2) at (0,-3.464)   {};
		\node[big dot] (B) at (-0.5,-0.866)   {};
		\node[big dot] (L3) at (1,-3.464)   {};
		
		\draw[-] (origin) -- (L1);
		\draw[-] (A) -- (L2);
		\draw[-] (B) -- (L3);
		
		\draw (origin) node[left]{$\rho$};
		\draw (A) node[left]{$v_1$};
		\draw (B) node[left]{$v_2$};
		\end{tikzpicture}
	\end{minipage}
	\begin{minipage}{.29\textwidth}
		\begin{tikzpicture}
		\node[big dot] (origin) at (0,0) {};
		\node[big dot] (A) at (-1,-1.732) {};
		\node[big dot] (L1) at (-2,-3.464)   {};
		\node[big dot] (L2) at (0,-3.464)   {};
		\node[big dot] (B) at (-0.5,-0.866)   {};
		\node[big dot] (L3) at (1,-3.464)   {};
		\node[big dot] (C) at (-0.25,-0.433)   {};
		\node[big dot] (L4) at (1.5,-3.464)   {};
		
		\draw[-] (origin) -- (L1);
		\draw[-] (A) -- (L2);
		\draw[-] (B) -- (L3);
		\draw[-] (C) -- (L4);
		
		\draw (origin) node[left]{$\rho$};
		\draw (A) node[left]{$v_1$};
		\draw (B) node[left]{$v_2$};
		\draw (C) node[left]{$v_3$};
		\end{tikzpicture}
	\end{minipage}
	\begin{minipage}{.29\textwidth}
		\begin{tikzpicture}
		\node[big dot] (origin) at (0,0) {};
		\node[big dot] (A) at (-1,-1.732) {};
		\node[big dot] (L1) at (-2,-3.464)   {};
		\node[big dot] (L2) at (0,-3.464)   {};
		\node[big dot] (B) at (-0.5,-0.866)   {};
		\node[big dot] (L3) at (1,-3.464)   {};
		\node[big dot] (C) at (-0.25,-0.433)   {};
		\node[big dot] (L4) at (1.5,-3.464)   {};
		\node[big dot] (D) at (-0.125,-0.2165)   {};
		\node[big dot] (L5) at (1.75,-3.464)   {};
		\node[big dot] (E) at (-0.0625,-0.10825)   {};
		\node[big dot] (L6) at (1.875,-3.464)   {};
		
		\draw[-] (origin) -- (L1);
		\draw[-] (A) -- (L2);
		\draw[-] (B) -- (L3);
		\draw[-] (C) -- (L4);
		\draw[-] (D) -- (L5);
		\draw[-] (E) -- (L6);
		
		\draw (origin) node[left]{$\rho$};
		\draw (A) node[left]{$v_1$};
		\draw (B) node[left]{$v_2$};
		\draw (C) node[left]{$v_3$};
		\end{tikzpicture}
	\end{minipage}
	\begin{minipage}{.10\textwidth}
		$\cdots\cdots$
	\end{minipage}
	\caption{A sequence of trees $\{T^k\}_k$ (from left to right) satisfying the big bang condition. The distance from $v_k$ to the root is $2^{-k}$.}\label{Fig:BigB}
\end{figure}

\noindent Finally we are ready to state our main combinatorial assumption.
\begin{assumption}[Taxon-rich setting: big bang condition]\label{A:T^k}
	We assume that
	$\{T^k\}_{k}$  
	\begin{enumerate}
	    \item 	is  a  nested sequence of trees  with common root $\rho$;
	    \item  has uniformly bounded height; and
	    \item satisfies the big bang condition.
	\end{enumerate}
	
\end{assumption}
\noindent 
We record the following consequence of~\cite[Theorem 1]{MR3814241} (see Section~\ref{S:ideas} below):
\begin{quote}
Under Assumption~\ref{A:T^k}, there exists a sequence of
root estimators that is consistent for the TKF91 process on $\{T^k\}_k$.
\end{quote}
However, this result is essentially existential
and that sequence of estimators is in general not directly computable (see Section 5.2 in~\cite[Theorem 1]{MR3814241} for the details
of the estimator).
The main contribution here is the design of a sequence of estimators which are not only {\it consistent} but also {\it explicit and computationally tractable}.  This is a first step towards designing practical estimators
with consistency guarantees.

\paragraph{Important simplification} Without loss of generality, we can assume that all leaves are at height $h^*$. This can be achieved artificially by
simulating the TKF91 process from 
the original leaf sequences up to
time $h^*$ and using the output as the new leaf sequences. We make this assumption throughout the rest of the paper, that is, from now on
$$
\ell_v = h^*, \qquad\forall v\in \partial T^k, \forall k.
$$
Alternatively, one could make use of the conditional law of the sequences at height $h^*$ given the sequences at the leaves of $T^k$, however this would unnecessarily complicate the presentation of our estimator.

\subsection{Main results}
\label{S:results}


In our main result, we devise and analyze a sequence of explicit, consistent root estimators for the TKF91 model under Assumption~\ref{A:T^k}. We first describe the root reconstruction algorithm.

\paragraph{Root reconstruction algorithm}
The {\bf input data} are the mutation parameters $(\nu,\,\lambda,\,\mu)\in (0,\infty)^3$ and $(\pi_A,\,\pi_T,\,\pi_C,\,\pi_G)\in [0,\infty)^4$ with  $\pi_A+\pi_T+\pi_C+\pi_G=1$ and,
for $k\geq 1$, the tree $T^k$ together with the leaf states $\{\vec{X}_{v}:\,v\in \partial T^{k}\}$.
Our root reconstruction algorithm has three main steps: to control correlations among leaves, we follow~\cite{MR3814241} and extract a ``well-spread'' restriction of $T^k$ (Step 1);  then, using only the
leaf states on this restriction, we estimate the root sequence length (Step 2) and finally the root sequence itself (Step 3). 	
The intuition behind the construction of the estimator is discussed in Section~\ref{S:ideas}.

We fix $t_j=j$ for all $j$ (although we could use any $0 \leq t_1 < t_2 < \cdots < +\infty$\footnote{The choice of $t_j$'s affects the constants in our quantitative bounds, but we have not tried to optimize them in the current work.}).
The following functions will help simplify some expressions: 
\begin{align*}
\eta(t) &=
\frac{1- e^{(\lambda-\mu )t}}{1-\gamma e^{(\lambda-\mu)t}},\\
\phi(t)&=\frac{(1-\gamma\eta(t))\,\big[ e^{-\mu t}  
(1-e^{-\nu t})+ \big(1-e^{-\mu t} - \eta(t)\big)\big]}{1-\eta(t)},\\
\psi(t)&= \big(1-\gamma\eta(t)\big)\,e^{-(\mu+\nu)t},
\end{align*}
where $\gamma=\lambda/\mu$.
Note that the function $\eta(t)$ is in fact the probability that a nucleotide dies and has no descendant at time $t$; see~\cite{thorne1991evolutionary}. 
For simplicity, we write 
\begin{equation}\label{Def:func_j}
\eta_j=\eta(h^*+t_j),\,\phi_j= \phi(h^*+t_j) \text{ and } \psi_j=\psi(h^*+t_j),
\end{equation}
where recall that $h^*$ was defined in~\eqref{Uniformh}.
We use the notation $[[x]]$ for the unique integer such that $[[x]]-1/2 \leq x < [[x]]+1/2$.
Finally, we let $T^k_{[z]}$ be the subtree of $T^k$ rooted at $z$.

\bigskip

\noindent {\bf Root Estimator}

Our estimator $F_{k,s}$ will  depend on a fixed chosen $s\in (0,h^*)$ and on the index $k$ of the tree.

\begin{itemize}
	\item {\bf Step 1: Restriction} 
	\begin{itemize}
		\item[-] Fix $s\in (0,h^*)$ and denote $\partial T^k(s) = \{z_1,\ldots,z_m\}$,  where $m=|\partial T^k(s)|$.
		
		\item[-] For each $z_i$, pick an arbitrary leaf $x_i \in \partial T^k_{[z_i]}$. 
		
		\item[-] Set $\tilde{T}^{k,s}$ to be the
		restriction of $T^k$ to $\{x_1,\ldots,x_m\}$.  See Figure \ref{Fig:Subtree} where $m=3$.
	\end{itemize}

	\item {\bf Step 2: Length estimator}
	\begin{itemize}
		\item Compute the root sequence-length estimator 
	\begin{equation}\label{Def:M^k*}
	 \widehat{M}^{\,k} =\Bigg[\Bigg[\,\frac{1}{m} \sum_{v\in \partial \tilde{T}^{k,s}}\bigg(	|\vec{X}_v|e^{(\mu-\lambda)h^*}-\frac{\lambda}{\mu-\lambda} (e^{(\mu-\lambda)h^*}-1) \bigg)\,\Bigg]\Bigg].
	\end{equation}
	
	\end{itemize}

	\item {\bf Step 3: Sequence estimator}
	
	\begin{itemize}
		\item Compute the conditional frequency estimator
		\begin{equation}\label{Def:freq}
		f^{k,s}_{\sigma}(t_j)= \frac{1}{m}\,\sum_{v\in \partial \tilde{T}^{k,s}} p^{\sigma}_{\vec{X}_{v}}(t_j)
		\end{equation}
		for $1\leq j\leq \widehat{M}^{\,k}$ and $\sigma\in\{A,T,C,G\}$, where 
		\begin{equation}\label{E:Q^sigma}
		p^{\sigma}_{\vec{x}}(t)= \pi_{\sigma}\,\phi(t)\,\big[1-\big(\eta(t)\big)^{|\vec{x}|}\big]\,+\,\psi(t)\sum_{i=1}^{|\vec{x}|}1_{\{x_i=\sigma\}}\big(\eta(t)\big)^{i-1} \,+\,\pi_{\sigma}\gamma\,\eta(t),
		\end{equation}
		when $\vec{x}=(x_i)_{i=1}^{|\vec{x}|}$.
		
		\item Set  $U$ to be the $\widehat{M}^{\,k}\times 4$ matrix with entries 
		\begin{equation}
		U_{j, \sigma}=f^{k,s}_{\sigma}(t_j)- \pi_{\sigma}\,\phi_j\,\big[1-\eta_j^{\widehat{M}^{\,k}}\big]\,-\,\pi_{\sigma}\,\gamma\,\eta_j,\label{S:def:U}
		\end{equation}
		set $\Psi$ to be the $\widehat{M}^{\,k}\times \widehat{M}^{\,k}$ diagonal matrix whose diagonal entries are $\{\psi_j\}_{j=1}^{\widehat{M}^{\,k}}$; and set
		$V:=V_{t_1,\cdots,t_{\widehat{M}^{\,k}}}$ to be the $\widehat{M}^{\,k}\times \widehat{M}^{\,k}$ Vandermonde matrix 
		\begin{equation*}
		V_{t_1,\cdots,t_{\widehat{M}^{\,k}}}=
		\begin{pmatrix}
		1&\eta_1&\eta_1^2&\dots&\eta_1^{\widehat{M}^{\,k}-1}\\
		1&\eta_2&\eta_2^2&\dots&\eta_2^{\widehat{M}^{\,k}-1}\\
		&&\vdots\\
		1&\eta_{\widehat{M}^{\,k}}&\eta_{\widehat{M}^{\,k}}^2&\dots &\eta_{\widehat{M}^{\,k}}^{\widehat{M}^{\,k}-1}\\
		\end{pmatrix},
		\end{equation*}
        where recall that $\eta_j$, $\phi_j$
        and $\psi_j$ were defined in~\eqref{Def:func_j}.
        		
		\item Define $F_{k,s}\left(\vec{X}_{\partial T^k}\right)$   to be an element in $\{A,T,C,G\}^{\widehat{M}^{\,k}}$ such that the $i$-th site satisfies 
		\begin{equation}\label{Def:F^M_k}
		\left[F_{k,s}\left(\vec{X}_{\partial T^k}\right)\right]_i \in \arg\max\left\{ (V^{-1}\Psi^{-1} U)_{i,\sigma} : \sigma \in \{A,T,C,G\}\right\}.
		\end{equation}
		If there is more than one choice, pick one uniformly at random.
		
	\end{itemize}

	
\end{itemize}

\paragraph{Statement of results}
Finally, our main claim is Theorem~\ref{thm:consistent} below, which asserts that the root estimator we just described provides 
a consistent root estimator in the sense of Definition \ref{Def:Consistent}.
Recall the sequence space $\mathcal{S}$ defined in \eqref{S}.
\begin{theorem}\label{thm:consistent}
	Suppose $\{T^k\}_k$ satisfies Assumption \ref{A:T^k}.    
Let $(s_k)_k$ be any sequence such that $s_k > 0$ and $s_k \downarrow 0$ as $k \to +\infty$, and let $F_k:=F_{k,s_k}$ be as defined in~\eqref{Def:F^M_k}.
	Then $\{F_k\}_k$ is a sequence of
	{\it consistent} root estimators for $\{T^k\}_k$.
\end{theorem}
\noindent  In the more general context of continuous-time countable-space Markov chains on bounded-height nested trees~\cite{MR3814241}, consistent root estimators were shown to exist under the big bang condition of Assumption~\ref{A:T^k} when,
\emph{in addition}, the edge process satisfies initial-state identifiability, i.e., the state of the process at time $0$ is uniquely determined by the distribution of the state at any time $t > 0$.
(In fact, these conditions are essentially necessary; see~\cite{MR3814241} for details.)
Moreover, reversibility of the TKF91 process together with an observation of~\cite{MR3814241} implies that the TKF91 process does indeed satisfy initial-state identifiability.

In that sense, Theorem~\ref{thm:consistent} is not new. However, the root estimators implicit in~\cite{MR3814241} have a major drawback from an algorithmic point of view. They rely on the computation of the total variation distance between the leaf distributions conditioned on different root states---and we are not aware of a tractable way to do these computations in the TKF91 model. 
In our main contribution, we give explicit root estimators that are statistically consistent under Assumption~\ref{A:T^k}. Specifically, our main novel contribution lies in Step 3 above, which is based on the derivation of explicit formulas to invert the mapping from the root sequence to the distribution of the leaf sequences. Moreover our estimator is computationally efficient in that the number of arithmetic operations needed scales like a polynomial of the total input sequence length. See Section~\ref{S:ideas} for an overview.

In a second novel contribution, we also derive a quantitative error bound.  Throughout this paper, we let $\P^{\vec{x}}$ be the probability measure when the root state is $\vec{x}$. If the root state is chosen according to a distribution $\Pi$, then we denote the probability measure by $\P^{\Pi}$. Finally, we denote by $\P_M$ the conditional probability measure for the event that the root state has length $M$.
\begin{theorem}[Error bound]\label{thm:error}
	Suppose $\{T^k\}_k$ satisfies Assumption~\ref{A:T^k} and $\{F_{k,s}\}_k$ are the root estimators described in~\eqref{Def:F^M_k}. 
	Then, for any $\epsilon\in (0,\infty)$, there exist positive constants $\{C_i\}_{i=1}^3$ such that
	\begin{align}
	\P^{\Pi}\left[F_{k,s}(\vec{X}^k_{\partial T^{k}})\neq \vec{X}^k_{\rho}\right] 
	&\leq \epsilon +  C_1\,\exp{\left(-C_2\,|\partial T^{k}(s)|\right)}\,+\,C_3\,s \label{error2}
	\end{align}
	for all $s\in (0,h^*/2]$ and $k\geq 1$. 
\end{theorem}
\noindent
More concretely if we seek an error probability of at most, say, $3\epsilon$  where $\epsilon >0$, then we pick 
\begin{equation*}
s< \min{\left\{\frac{\epsilon}{C_3},\, \frac{h^*}{2}\right\}}
\end{equation*}
and construct root estimator $F_{k,s}$ as described in~\eqref{Def:F^M_k}. Theorem \ref{thm:error}  guarantees that
$$\P^{\Pi}\left[F_{k,s}(\vec{X}^k_{\partial T^{k}})\neq \vec{X}^k_{\rho}\right] \leq 3\epsilon$$ 
whenever $k$ is large enough such that
\begin{equation*}
|\partial T^k(s)| \geq \frac{1}{C_2}\ln\left(\frac{C_1}{\epsilon} \right).
\end{equation*}

\section{Key ideas in the construction of the root estimator}\label{S:ideas}

\noindent In our main contribution, we give explicit, computationally efficient root estimators that are consistent under Assumption~\ref{A:T^k}. The estimators  were defined in Section~\ref{S:results}. In this section, we motivate these estimators by giving an alternative, constructive proof of initial-state identifiability in the TKF91 process. At a high level, Theorems~\ref{thm:consistent} and~\ref{thm:error} are then established along the following lines: the big bang condition ensures the existence of a large number of leaves whose states are ``almost independent'' conditioned on the root state, and concentration arguments imply that the sample version of the inversion formula derived in Lemma~\ref{T:Initial} below is close to the root state. See Section~\ref{S:proofs} for details.

The key ideas are encapsulated in the following lemmas about the \textit{edge process}---that is, there is no tree yet at this point of the exposition.
Let  $\E_{M}[|\mathcal{I}_{t}|]$ be the expected length of a sequence $\mathcal{I}$ after running the TKF91 edge process for time $t$, starting from a sequence of length $M$.  Recall the definition of $p^{\sigma}_{\vec{x}}(t)$ in~\eqref{E:Q^sigma}.
We will show in Lemma \ref{L:theta_v} that $p^{\sigma}_{\vec{x}}(t)$ is equal to the probability of observing the nucleotide $\sigma \in \{A,T,C,G\}$ as the \textit{first nucleotide} of a sequence at time $t$, under the TKF91 edge process with initial state $\vec{x}$. 
Formally, let $\mathcal{I}_t(1)$ denote the \textit{first nucleotide} of the sequence $\mathcal{I}$ at time $t$ and let $\{\mathcal{I}_t(1) = \sigma\}$ be the event that this first nucleotide is $\sigma$.		
\begin{lemma}[Distribution of the first nucleotide]\label{L:theta_v}
	The probability that $\sigma$ is the first nucleotide at time $t$ is
	\begin{equation}\label{Def:theta_v}
	 \P^{\vec{x}}\left(\bf\mathcal{I}_t(1) = \sigma\right)
	= p^{\sigma}_{\vec{x}}(t),
	\end{equation}
	for all  $\sigma \in\{A,T,C,G\}$, sequence $\vec{x}\in \mathcal{S}$ and time $t\in(0,\infty)$, where $p^{\sigma}_{\vec{x}}(t)$ is defined in~\eqref{E:Q^sigma}.
\end{lemma}
\begin{proof}
	We use some notation of~\cite{thorne1991evolutionary}, where a nucleotide is also referred to as a ``normal link'' and a generic such nucleotide is denoted $``\star"$. We define
	\begin{align*}
	p_k:=p_k(t)&:=\P_{\star}( \text{normal link }``\star" \text{ survives and has }k\text{ descendants including itself}),\\
	p^{(1)}_k :=p^{(1)}_k(t) &:=\P_{\star}( \text{normal link }``\star" \text{ dies and has }k\text{ descendants}),\\
	p^{(2)}_k :=p^{(2)}_k(t)&:=\P_{\bullet}(\text{immortal link }``\bullet" \text{ has }k\text{ descendants including itself}).
	\end{align*}
	These probabilities are explicitly found in  \cite[Eq.~(8)-(10)]{thorne1991evolutionary} by solving the differential equations governing the underlying birth and death processes:
	\begin{align*}
	&\Bigg\{\begin{array}{ll}
	p_n(t) &= e^{-\mu t}(1-\gamma\eta(t))[\gamma\eta(t)]^{n-1}\\
	p^{(1)}_n(t) &= (1-e^{-\mu t}-\eta(t))(1-\gamma\eta(t))[\gamma\eta(t)]^{n-1}\\
	p_n^{(2)}(t) &= (1-\gamma\eta(t))[\gamma\eta(t)]^{n-1}
	\end{array}\qquad \text{for }n\ge 1, 
	\\
	&\Bigg\{\begin{array}{ll}
	p_0(t) &=0\\
	p^{(1)}_0(t) &= \eta(t)\\
	p_0^{(2)}(t) &= 0.
	\end{array}
	\end{align*}
	
	Let $\mathcal{K}_{\bullet}$ be the event that the first nucleotide $\mathcal{I}_t(1)$ is the descendant of the immortal link $``\bullet"$ and let
	$\mathcal{K}_i$ be the event that the first nucleotide is the descendant of $v_i$ for  $1\leq i\leq |\vec{v}|$. By the law of total probability,
	\begin{align}
	&\P^{\vec{v}}\left(\mathcal{I}_t(1) = \sigma\right)
	= \P^{\vec{v}}(\mathcal{I}_t(1) = \sigma,\,\mathcal{K}_{\bullet}) + \sum_{i=1}^{|\vec{v}|} \P^{\vec{v}}(\mathcal{I}_t(1) = \sigma,\,\mathcal{K}_i).\label{theta_v3}
	\end{align}
	We now compute each term on the RHS. In the rest of the proof of Lemma~\ref{L:theta_v}, to simplify the
	notation we set $\eta:=\eta(t)$.
	
	We have that $ \P^{\vec{v}}(\mathcal{I}_t(1) = \sigma\,|\,\mathcal{K}_{\bullet})= \pi_{\sigma}$, because any descendant of the immortal link corresponds to an insertion and any inserted nucleotide is independent of the other variables with distribution $(\pi_A,\pi_T,\pi_C,\pi_G)$. We note that $\P^{\vec{v}}(\mathcal{K}_{\bullet})=1-p^{(2)}_1= \gamma\eta$,
	because $\mathcal{K}_{\bullet}$ occurs if and only if $``\bullet"$ has at least two descendants including itself.
	Hence 
	\begin{equation}\label{theta_v3_1}
	    \P^{\vec{v}}(\mathcal{I}_t(1) = \sigma,\,\mathcal{K}_{\bullet})=\pi_{\sigma}\,\gamma\eta,
	\end{equation}

    For $1\leq i\leq |\vec{v}|$, $\mathcal{K}_i$ is the event that the normal link $v_i$ either survives or dies but has at least $1$ descendant, while the offspring of all previous links die.
    Further, let $S_i$ be the event that $v_i$ survives, which has probability $e^{-\mu t}$. Then
    \begin{equation*}
        \P^{\vec{v}}(\mathcal{K}_i \cap S_i)=p^{(2)}_1(p^{(1)}_0)^{i-1}e^{-\mu t}=(1-\gamma \eta)\eta^{i-1}\,e^{-\mu t},
    \end{equation*}
    since  $p^{(2)}_1$ is the probability that the immortal link has exactly 1 offspring which is itself, $(p^{(1)}_0)^{i-1}$ is the probability that all $\{v_j\}_{j=1}^{i-1}$ were deleted and left no descendant
    and $S_i$ is independent of the previous events.
	Moreover, we have $\P^{\vec{v}}(\mathcal{I}_t(1) = \sigma\,|\,\mathcal{K}_i \cap S_i)= f_{v_i \sigma}$, where
	\begin{equation*}
	f_{ij}:=f_{ij}(t) 
	=\pi_{j}(1-e^{-\nu t})+e^{-\nu t}1_{\{i=j\}}
	\end{equation*}
	is the transition probability that a nucleotide is of type $j$ after time $t$, given that it is of type $i$ initially. (Recall from Definition~\ref{Def_TKF91} that
	$\nu>0$ is the substitution rate.)
    Letting $S_i^c$ be the complement of $S_i$, then
     \begin{equation*}
      \P^{\vec{v}}(\mathcal{K}_i \cap S_i^c)=(1-\gamma \eta)\eta^{i-1}\, \left(\sum_{k\geq 1}p^{(1)}_k \right) =(1-\gamma \eta)\eta^{i-1}\, (1-e^{-\mu t}-\eta),
    \end{equation*}
    and $\P^{\vec{v}}(\mathcal{I}_t(1) = \sigma\,|\,\mathcal{K}_i \cap S^c_i)=\pi_{\sigma}$. Therefore,
    \begin{align}
    &\P^{\vec{v}}(\mathcal{I}_t(1) = \sigma\,,\,\mathcal{K}_i) \notag\\
    &=\P^{\vec{v}}(\mathcal{I}_t(1) = \sigma\,|\,\mathcal{K}_i \cap S_i)\,\P^{\vec{v}}(\mathcal{K}_i \cap S_i)+\P^{\vec{v}}(\mathcal{I}_t(1) = \sigma\,|\,\mathcal{K}_i \cap S^c_i)\,\P^{\vec{v}}(\mathcal{K}_i \cap S_i^c)   \notag\\
    &= (1-\gamma \eta)\eta^{i-1}\,\left( f_{v_i \sigma}\,e^{-\mu t} + \pi_{\sigma} (1-e^{-\mu t}-\eta) \right). \label{theta_v3_2}
    \end{align}

	Putting \eqref{theta_v3_1} and \eqref{theta_v3_2} into \eqref{theta_v3}, 
	\begin{align*}
    \P^{\vec{v}}\left(\mathcal{I}_t(1) = \sigma\right)
    &= \pi_{\sigma}\gamma\eta
    + \sum_{i=1}^{|\vec{v}|}
    \left[
    (1-\gamma \eta)\eta^{i-1}\,\left( f_{v_i \sigma}\,e^{-\mu t} + \pi_{\sigma} (1-e^{-\mu t}-\eta) \right)
    \right]\\
	&= (1-\gamma\eta) e^{-\mu t} \left[\sum_{i=1}^{|\vec{v}|} \eta^{i-1}f_{v_i \sigma}\right]+ (1-\gamma\eta)(1-e^{-\mu t} -\eta)\pi_{\sigma} \frac{1-\eta^{|\vec{v}|}}{1-\eta} + \pi_{\sigma} \gamma\eta,	
	\end{align*}
	which is exactly \eqref{Def:theta_v} upon further re-writing
	\begin{equation*}
	\sum_{i=1}^{|\vec{v}|}\eta^{i-1}f_{v_i \sigma}
	=\pi_{\sigma}(1-e^{-\nu t})\frac{1-\eta^{|\vec{v}|}}{1-\eta}+e^{-\nu t}\sum_{i=1}^{|\vec{v}|}1_{\{v_i=\sigma\}}\eta^{i-1}.
	\end{equation*}
	The proof of Lemma \ref{L:theta_v} is complete.\end{proof}

Building on Lemma \ref{L:theta_v}, our second lemma regarding the edge process gives
a constructive proof of initial-state identifiability. Recall that we set $t_j = j$.
\begin{lemma}[Constructive proof of initial-state identifiability]\label{T:Initial}
	For any $h^* \geq 0$, the following mappings are one-to-one and  have explicit expressions for their inverses:
	\begin{itemize}
		\item[(i)] the mapping $\Phi_{h^*}^{(1)}:\,\ZZ_+\to \RR_+$  defined by 
		\begin{equation*}
		\Phi_{h^*}^{(1)}(M)=\E_{M}[|\mathcal{I}_{h^*}|],
		\end{equation*}
		\item[(ii)] the mappings
		$\Phi^{(2)}_{h^*+t_1,\cdots,h^*+t_{M}}:\,\{A,T,C,G\}^M\to [0,1]^{4M}$  defined by 	
		\begin{equation}\label{Phi^2}
		\Phi^{(2)}_{h^*+t_1,\cdots,h^*+t_{M}}(\vec{x})= \left(p^{\sigma}_{\vec{x}}(h^*+t_j)\right)_{\sigma \in \{A,T,C,G\},\,1\leq j\leq M},
		\end{equation}
		for any $\vec{x}$ with $|\vec{x}| = M \geq 1$.
	\end{itemize}	 	 
\end{lemma}
\begin{proof}
(i) The sequence length of the TKF91 edge process is a
well-studied stochastic process known as a continuous-time linear birth-death-immigration process $(|\mathcal{I}_{t}|)_{t\geq 0}$.
(We give more details on this process in Section~\ref{sec:tkf-basic}.) 
The expected sequence length, in particular, is known to satisfy the following differential
equation
$$
\frac{d}{dt} \E_{M}[|\mathcal{I}_t|]
= -(\mu-\lambda) \E_{M}[|\mathcal{I}_t|] + \lambda,
$$
with initial condition $\E_{M}[|\mathcal{I}_0|] = M$,
whose solution is
given by
\begin{equation}\label{Elength2}
\E_{M}[|\mathcal{I}_t|] =
\Phi_t^{(1)}(M)
:=
M\beta_t+\frac{\gamma(1-\beta_t)}{1 -\gamma},
\end{equation}
where $\beta_t = e^{-(\mu-\lambda)t}$ and $\gamma =\frac{\lambda}{\mu}$. Solving~\eqref{Elength2} for $M$, we see that $\Phi_t^{(1)}$ is injective with inverse 
	\begin{equation}\label{InversePhi1}
	(\Phi_t^{(1)})^{-1}(z)=\frac{1}{\beta_{t}}\left(z-\frac{\gamma(1-\beta_{t})}{1 -\gamma}\right).
	\end{equation}

\bigskip

(ii) For $\Phi^{(2)}$, we 
make the following claim regarding
	the probability of observing $\sigma$ as the \textit{first nucleotide} of a sequence at time $t$ under the TKF91 edge process starting at $\vec{x}$ with $|\vec{x}|\geq 1$.

Recalling the functions $\phi_j$ and $\psi_j$ defined in \eqref{Def:func_j}, we have, by Lemma \ref{L:theta_v}, 
\begin{equation}\label{E:Solve for x}
p^{\sigma}_{\vec{x}}(h^*+t_j)= \pi_{\sigma}\,\phi_j\,\big[1-\eta_j^{M}\big]\,+\,\psi_j\sum_{i=1}^{M}1_{\{x_i=\sigma\}}\,\eta_j^{i-1} \,+\,\pi_{\sigma}\gamma\,\eta_j,
\end{equation}
for all $\sigma\in \{A,T,C,G\}$  and $1\leq j\leq M$,
where $M = |\vec{x}|$. We then solve~\eqref{E:Solve for x} for $\vec{x}=(x_i)\in \{A,T,C,G\}^M$.
System~\eqref{E:Solve for x} is equivalent to the matrix equation 
	\begin{equation}\label{E:Solve for x2}
	H = \Psi\,V\,Y^{\vec{x}} 
	\end{equation}
	where $\Psi$ and $V$ are the $M\times M$ matrices defined in Section~\ref{S:results}, and
	\begin{enumerate}
		\item 	$Y^{\vec{x}}$ is the $M\times 4$ matrix whose entries are  $1_{\{x_j=\sigma\}}$, and
		\item $H$ is the $M\times 4$ matrix with entries 
		\begin{equation*}
		H_{j, \sigma}=p^{\sigma}_{\vec{x}}(h^*+t_j)- \pi_{\sigma}\,\phi_j\,\big[1-\eta_j^{M}\big]\,-\,\pi_{\sigma}\gamma\,\eta_j.
		\end{equation*}
	\end{enumerate}
	It is well-known that the Vandermonde matrix $V$ is invertible (see, e.g.,~\cite[Theorem 1]{gautschi1962inverses}), so we can solve the system~\eqref{E:Solve for x2} to obtain
		\begin{equation}\label{E:Solve for x3}
		Y^{\vec{x}}= V^{-1}\,\Psi^{-1}\,H.
		\end{equation}
	Sequence $\vec{x}\in \{A,T,C,G\}^{M}$ is uniquely determined by $ Y^{\vec{x}}$. Hence, from~\eqref{E:Solve for x3}, we get an explicit inverse for the mappings $\Phi^{(2)}_{h^*+t_1,\cdots,h^*+t_{M}}$ defined in~\eqref{Phi^2}. 
	
	The proof of Lemma \ref{T:Initial} is complete.
\end{proof}

Heuristically, Steps 2 and 3 in the root estimator defined in Section~\ref{S:results} are the sample versions of the inverses of 
$\Phi_{h^*}^{(1)}$ and $\Phi^{(2)}_{h^*+t_1,\cdots,h^*+t_{M}}$.
That is, we replace the expectation $\E_{|\vec{x}|}[|\mathcal{I}_{h^*}|]$ and probabilities $\left(p^{\sigma}_{\vec{x}}(h^*+t_j)\right)_{\sigma \in \{A,T,C,G\},\,1\leq j\leq M}$ with their corresponding empirical averages--- conditioned on the 
observations at the leaves. The reason we
consider several ``future times'' $h^*+t_1,\cdots,h^*+t_{M}$ is to ensure that we have
sufficiently many equations to ``invert the system'' and obtain the full root sequence, as described in Lemma \ref{T:Initial} (ii) and Step 3 of the root estimator.

\section{Proofs}
\label{S:proofs}

To prove Theorem \ref{thm:error} (and Theorem~\ref{thm:consistent} from which it follows), that is, to obtain the desired upper bound, ~\eqref{error2}, on
\begin{equation}\label{Pi to x}
\P^{\Pi}\left[F_{k,s}(\vec{X}^k_{\partial T^{k}})\neq \vec{X}^k_{\rho}\right]=\sum_{\vec{x}\in \mathcal{S}}\P^{\vec{x}}\left[F_{k,s}(\vec{X}^k_{\partial T^{k}})\neq \vec{x}\right]\,\Pi(\vec{x}),
\end{equation}
we first observe that from the construction of $F_{k,s}$ in Section~\ref{S:results}, if the estimator is wrong, then either the estimated length is wrong or the length is correct but the sequence is wrong. 
Let $F^{(M)}_{k,s}$ denote the estimator in~\eqref{Def:F^M_k} when $\widehat{M}^{\,k} = M$. We have 
\begin{align}
\P^{\vec{x}}\left[F_{k,s}(\vec{X}^k_{\partial T^{k}})\neq \vec{x}\right] 
&= 
\P^{\vec{x}}\left[\{F_{k,s}(\vec{X}^k_{\partial T^{k}})\neq \vec{x}\}\cap \{\widehat{M}^{\,k}\neq |\vec{x}|\}\right]\nonumber\\ 
&\qquad +\P^{\vec{x}}\left[\{F_{k,s}(\vec{X}^k_{\partial T^{k}})\neq \vec{x}\}\cap \{\widehat{M}^{\,k}= |\vec{x}|\}\right]\nonumber\\
&\leq 
\P_{|\vec{x}|}\left[\widehat{M}^{\,k}\neq |\vec{x}|\right] +\P^{\vec{x}}\left[F^{(|\vec{x}|)}_{k,s}(\vec{X}^k_{\partial T^{k}})\neq \vec{x}\right],\label{Pf_Fk}
\end{align}
for all $\vec{x}\in \mathcal{S}\setminus \{``\bullet"\}$, 
where recall that $\P^{\vec{x}}$ denotes the probability law when the root state is $\vec{x}$ and $\P_M$ denotes the probability law when the root state has length $M$. The proof is therefore reduced to bounding the first and second terms on the RHS of~\eqref{Pf_Fk},  which are formulated as Propositions~\ref{P:Length} and~\ref{P:Seq} respectively in the next two subsections. To simplify
the notation, throughout this section
we let $\vec{X} := \vec{X}^k$.



\subsection{Reconstructing the sequence length}

In this subsection, we bound the first term on the RHS of \eqref{Pf_Fk}, which is the probability of incorrect estimation of the root sequence length.
Proposition \ref{P:Length} is a quantitative statement about how this error tends to zero exponentially fast in $m:=|\partial T^{k}(s)|$, as $k\to\infty$ and $s\to 0$. 
\begin{prop}[Sequence length estimation error]\label{P:Length}
	There exist constants $C_1,\,C_2\in (0,\infty)$ which depend only on  $h^*$,  $\mu$ and $\lambda$ such that
	\begin{align*}\label{E:Length}
    \P_M(\widehat{M}^{\,k}\neq M) &\leq  2\exp{\left(-C_1\,|\partial T^{k}(s)|\right)}\,+\,C_2\,(M+1)\,s
	\end{align*}
	for all $s\in (0,h^*/2]$, $k\in \mathbb{N}$ and $M\in \mathbb{N}$. 
\end{prop}

\bigskip

\paragraph{Outline of proof:}
Fix $k\geq 1$, $s>0$, and $m:=|\partial T^{k}(s)|$. Recall that
$$
\beta_t = e^{-(\mu-\lambda)t} \qquad \text{and}  \qquad \gamma =\frac{\lambda}{\mu},
$$
and that the definition of the length estimator is
$$
\widehat{M}^{\,k} =\Bigg[\Bigg[\,\frac{1}{m} \sum_{v\in \partial \tilde{T}^{k,s}}\bigg(	|\vec{X}_v|e^{(\mu-\lambda)s}+\frac{\lambda}{\mu-\lambda} (1-e^{(\mu-\lambda)s}) \bigg)\,\Bigg]\Bigg],
$$
which (up to rounding) is a linear combination of the sequence lengths at the leaves of the restriction 
$\tilde{T}^{k,s}$. To explain this expression,
we note first that the expected sequence length after running the edge process for time $h^*$
is a function of the sequence length $M$
at the root, specifically 
$\Phi_{h^*}^{(1)}(M)$,
where we used the notation defined in Lemma~\ref{T:Initial}. We can estimate the latter from the sequence lengths at the leaves of the restriction using the
empirical average
$$
A^{k,s}=\frac{1}{m} \sum_{v\in \partial \tilde{T}^{k,s}} |\vec{X}_v|
$$ 
whose expectation is
$\Phi_{h^*}^{(1)}(M)$.
Lemma~\ref{T:Initial} then allows us to recover 
$M$ approximately by inversion and rounding. Specifically, by linearity of $\Phi^{(1)}$ and Equations~\eqref{Elength2} and~\eqref{InversePhi1},
\begin{align*}
(\Phi_{h^*}^{(1)})^{-1}(A^{k,s})
&= (\Phi_{h^*}^{(1)})^{-1}
\left(\frac{1}{m} \sum_{v\in \partial \tilde{T}^{k,s}} |\vec{X}_v|
\right)\\
&= \frac{1}{m} \sum_{v\in \partial \tilde{T}^{k,s}} (\Phi_{h^*}^{(1)})^{-1}
\left(|\vec{X}_v|
\right)\\
&= \frac{1}{m} \sum_{v\in \partial \tilde{T}^{k,s}}
\frac{1}{\beta_{h^*}}\left(
|\vec{X}_v|
-
\frac{\gamma (1 - \beta_{h^*})}{1-\gamma}
\right)\\
&= \frac{1}{m} \sum_{v\in \partial \tilde{T}^{k,s}}
\left(
|\vec{X}_v|e^{(\mu-\lambda)h^*}-\frac{\lambda}{\mu-\lambda} \left(e^{(\mu-\lambda)h^*} - 1\right)
\right),
\end{align*}
Finally, we see that
$\widehat{M}^{\,k}
= [[(\Phi_{h^*}^{(1)})^{-1}(A^{k,s})]]$.
We show below that $A^{k,s}$ is concentrated around its mean
$\Phi_{h^*}^{(1)}(M)$, from which we
deduce that $\widehat{M}^{\,k}$
is concentrated around $M$. One technical
issue is to control the correlation
between the sequence lengths at the leaves
of the restriction. See below for the
details.

We first state as a lemma the above
observations. Throughout this section,
we use a generic sequence-length  process $(|\mathcal{I}_{t}|)_{t\geq 0}$ defined on a separate probability space from $\vec{X}^k$. Its purpose is to simplify various expressions involving the law
of sequence lengths.
\begin{lemma}[Concentration of $A^{k,s}$ suffices]
\begin{equation}\label{E:Length2}
\P_M(\widehat{M}^{\,k}\neq M)\leq \P_M\Big(\Big|A^{k,s}-\nu_M\Big| \geq \frac{\beta_{h^*}}{2}\Big),
\end{equation}
for all $M\geq 0$, where 
\begin{equation}\label{Def:nu*}
\nu_M:=\E_M[|\mathcal{I}_{h^*}|]=M\beta_{h^*}+\frac{\gamma(1-\beta_{h^*})}{1 -\gamma}.
\end{equation}
\end{lemma}

\begin{proof}
 All variables $|\vec{X}_v|$ for $v\in \partial \tilde{T}^{k,s}$ have the same mean $\nu_M$, where formula \eqref{Def:nu*} follows from \eqref{Elength2}. Formula \eqref{Elength2} also tells us that the function $M\mapsto \Phi_{h^*}^{(1)}(M)$ is linear with slope $\beta_{h^*}$. So if $|A^{k,s}-\nu_M| < \beta_{h^*}/2$ then the effect of the rounding is such that
\begin{equation*}
	[[\,	(\Phi_{h^*}^{(1)})^{-1}(A^{k,s})\,]]=[[\,	(\Phi_{h^*}^{(1)})^{-1}(\nu_M)\,]]=M.
\end{equation*}
Hence, the error probability satisfies \eqref{E:Length2}.   
\end{proof}

We move on to the formal proof,
which follows from a series
of steps. To bound the RHS of \eqref{E:Length2}, note that even though the variables $|\vec{X}_v|$, for $v\in \partial \tilde{T}^{k,s}$, are {\it correlated}, the construction of
$\tilde{T}^{k,s}$ guarantees that this correlation is ``small'' as the paths to the root from any two leaves of the restriction overlap only ``above level $s$.'' To control the correlation, 
we condition on the leaf states of the truncation $\vec{X}_{\partial T(s)}=\{\vec{X}_{v}:\,v\in \partial T(s)\}$. By the Markov property,
\begin{equation}\label{CondExpectation}
\E_{M}\left[A^{k,s}\middle|\vec{X}_{\partial T^k(s)}\,\right] = \frac{1}{m}\sum_{u\in \partial T^k(s)}\E_{|\vec{X}_u|}[|\mathcal{I}_{h^*-\,s}|].
\end{equation}
with $a \wedge b=\min\{a,b\}$, where note that the sum here (unlike $A^{k,s}$ itself) is over the leaves of the \emph{truncation at $s$}.
We first bound the probability that the conditional expectation \eqref{CondExpectation}
is close to its expectation $\nu_M$. 
Then, conditioned on this event, we establish concentration of $A^{k,s}$ used on conditional independence. 
We detail the above argument in Steps A-C as follows. Properties of the length process derived in Section~\ref{sec:tkf-basic} will be employed.
\paragraph{(A) Decomposition from conditioning on level $s$} 
We first decompose the RHS of \eqref{E:Length2} according to whether the sequence lengths at level $s$ are sufficiently concentrated.
For  $\epsilon>0$, $\delta\in(0,\infty)$ and $s\in (0,\,h^*)$, we have
\begin{equation}\label{E:Length2terms}
\P_M\left(\big|A^{k,s}-\nu_M\big| \geq \epsilon \right)
\leq \,\P_M\left(\{\big|A^{k,s}-\nu_M\big| \geq \epsilon\}\;\cap\; \mathcal{E}^c_{\delta,s}\right) + \P_{M}( \mathcal{E}_{\delta,s}),
\end{equation}
where
\begin{equation}\label{E_delta}
\mathcal{E}_{\delta,s}:= \Big\{\Big|\E_{M}\left[A^{k,s}\middle|\vec{X}_{\partial T^k(s)}\,\right]  -  \,\nu_M\Big|> \delta \Big\}.
\end{equation} 
is the event that the conditional expectation \eqref{CondExpectation}
away from its expectation $\nu_M$ by $\delta$.
As we shall see, Proposition \ref{P:Length} will be obtained by taking $\delta=\epsilon/2=\beta_{h^*}/4$.

\paragraph{(B) Bounding $\P_{M}( \mathcal{E}_{\delta,s})$}The second term on the RHS of \eqref{E:Length2terms} is treated in Lemma \ref{L:E_delta,s} below.
Because of the correlation above level $s$, we use Chebyshev's inequality to control the deviation of the conditional expectation.
\begin{lemma}[Conditional expectation given level $s$]\label{L:E_delta,s}
	For all $M\in  \ZZ_+$, $\delta\in (0,\infty)$ and $s\in (0,\,h^*)$, we have
	\begin{equation*}
	\P_{M}( \mathcal{E}_{\delta,s}) \leq  \delta^{-2}\,C\,(M+1)\,s,
	\end{equation*}
	where $C\in (0,\infty)$ is a constant which depends only on $h^*$, $\mu$ and $\lambda$. 
\end{lemma}
\begin{proof}
	By \eqref{CondExpectation} and Chebyshev's inequality, 
	\begin{align}
	\P_{M}(\mathcal{E}_{\delta,s})=& \P_M\left[\left|\sum_{u\in \partial T^k(s)}\E_{|\vec{X}_u|}[|\mathcal{I}_{h^*-s}|]  - m \nu_M\right|> m\delta \right] \notag\\
	\leq & \frac{1}{m^2\,\delta^2}\,\var_M\left[\sum_{u\in \partial T^k(s)}\E_{|\vec{X}_u|}[|\mathcal{I}_{h^*-s}|]\right]
	\label{E_delta2},
	\end{align}	
	where $\var_M$ is the variance conditioned on the initial sequence
	having length $M$.
    We will use the following simple
    identity for square-integrable
    variables
    \begin{equation}\label{S:CS-AM-GM}
        \var\left(
        \sum_{b=1}^B W_b
        \right)
        \leq
        B \sum_{b=1}^B \var(W_b)
    \end{equation}
    which follows from applying Cauchy-Schwarz to the covariance terms and then the AM-GM inequality.
	Then the variance above is bounded by
	\begin{align}
	\var_M\left[\sum_{u\in \partial T^k(s)}\E_{|\vec{X}_u|}[|\mathcal{I}_{h^*-s}|]\right]
	&\leq m^2 \,\var_M \left[\E_{|\vec{X}_u|}[|\mathcal{I}_{h^*-s}|]\right]\label{E_delta,s_1},
	\end{align}
	for any $u \in \partial T^k(s)$, where we used the fact that the $|\vec{X}_u|$'s for all
	such $u$'s are identically distributed.

By \eqref{Elength2},
\[
\E_{|\vec{X}_u|}[|\mathcal{I}_{h^*-s}|]=|\vec{X}_u|\,\beta_{h^*-s}+\frac{\gamma(1-\beta_{h^*-s})}{1 -\gamma}.
\]
So, plugging this last expression into the
variance formula~\eqref{Varlength'}, we obtain
\begin{align}
    &\var_M \left(\E_{|\vec{X}_u|}[|\mathcal{I}_{h^*-s}|]\right)\notag\\
    &=\var_M \left(|\vec{X}_u|\,\beta_{h^*-s}+\frac{\gamma(1-\beta_{h^*-s})}{1 -\gamma}\right)\notag\\
    &=\beta_{h^*-s}^2\,\var_M(|\vec{X}_u|) \notag\\
	&= \beta_{h^*-s}^2\,\left(M\,\frac{\beta_{s}(1-\beta_{s})(1+3\gamma-2\beta_{s})}{1-\gamma}\,+\,\frac{\gamma(1-\beta_{s})(1-\beta_{s}\gamma)}{(1-\gamma)^2}\right) \notag\\
	&\leq \beta_{h^*-s}^2\,(1-\beta_{s})\,C\,(M+1)\notag\\ 
	&\leq s\,C\,(M+1), \label{E_delta,s_2}
\end{align}
where $C\in (0,\infty)$ is a constant which depends only on $h^*$, $\mu$ and $\lambda$, where we used that $\gamma < 1$ and $\beta_t$ is decreasing in $t$.

Putting \eqref{E_delta,s_2} into  \eqref{E_delta,s_1},  the desired inequality follows from \eqref{E_delta2}.\end{proof}

\paragraph{(C) Deviation of $A^{k,s}$ conditioned on $\mathcal{E}_{\delta,s}^c$}
The remaining step of the analysis is to control the first term on the RHS of \eqref{E:Length2terms} by taking advantage of the conditional independence of the leaves of the restriction given $\vec{X}_{\partial T^k(s)}$. 
By the definition of $A^{k,s}$, 
the Markov property and conditional independence, that term is equal to
\begin{equation}\label{cond}
\sum_{\vec{u} \in \mathcal{G}_{\delta,s}^c} \, \P\left(\left|\sum_{i=1}^m L^{(u_i)}\;-\;m\nu_M\right| \geq m\epsilon \right) \,\P_M\left(\left|\vec{X}_{\partial T^k(s)}\right|=\vec{u}\right) ,
\end{equation}	
where the set $\mathcal{G}_{\delta,s}$ 
(closely related to $\mathcal{E}_{\delta,s}$) 
is defined as
$$\mathcal{G}_{\delta,s}:=\Big\{\vec{u}=(u_1,\cdots,u_m)\in \ZZ_+^m:\;\Big|\sum_{i}\E_{u_i}[|\mathcal{I}_{h^*-s}|]  - m \nu_M\Big| >  m\delta \Big\},$$
and $\{L^{(u_i)}\}_{i=1}^m$ are, under probability $\P$, {\it independent} copies of the birth-death process $|\mathcal{I}_{t}|$ run for time
$h^*-s$ starting from $\{u_i\}_{i=1}^m$ respectively. By the triangle inequality, we have for all $s\in (0,h^*)$,
\begin{align}
\eqref{cond}
&\leq  \max_{\vec{u} \in \mathcal{G}_{\delta,s}^c} \, \P\Big(\Big|\sum_{i=1}^m L^{(u_i)}\;-\;m\nu_M\Big| \geq m\epsilon \Big) \notag \\		
&\leq \max_{\vec{u} \in  \mathcal{G}_{\delta,s}^c}\P\left(\left|\sum_{i=1}^m \left(L^{(u_i)}-\E_{u_i}[|\mathcal{I}_{h^*-s}|]\right)\right| \geq m(\epsilon-\delta) \right)\notag\\
&\leq \max_{\vec{u} \in  \mathcal{G}_{\delta,s}^c}\P\left(\left|\sum_{i=1}^m \left(L^{(u_i)}-
\beta_{h^*-s} u_{i} - \frac{\gamma(1-\beta_{h^*-s})}{1 -\gamma}\right)\right| \geq m(\epsilon-\delta) \right). \label{cond2max}
\end{align}
Write $t=h^*-s$ for convenience. 

Each $L^{(u_i)}$ above can in fact be thought of as a sum of independent variables
$L'_i + \sum_{j=1}^u L''_{ij}$ where
$L'_i$ and $L_{ij}''$ are the sizes of the progenies of the immortal link and of the
$j$-th normal link respectively
in the edge process started with
sequence length $u$ and run for
time $t$. The expectations of $L'_i$
and $L''_{ij}$ are 
$\frac{\gamma(1-\beta_t)}{1 -\gamma}$
and
$\beta_t$
respectively (as can be seen from~\eqref{Elength2}). Both $L'_i$ and
$L''_{ij}$ have finite exponential moments
(see Section~\ref{sec:tkf-basic}) and hence, by
standard results in large deviations theory (see, e.g.,~\cite{Durrett:96}),
we have the bounds
$$
\P\left[\left|\sum_{i=1}^m L'_i 
- \frac{\gamma(1-\beta_t)}{1 -\gamma} m \right| \geq m(\epsilon-\delta)/2\right] \leq 
\exp\left(-C' m \right),
$$
and
$$
\P\left[\left|\sum_{i=1}^m \sum_{j=1}^{u_i} L''_{ij}
- \beta_t \sum_{i=1}^m u_i \right| \geq m(\epsilon-\delta)/2\right] \leq 
\exp\left(-C'' \sum_{i=1}^m u_i \right),
$$
where $C'$ and $C''$ are strictly
positive constants depending only on $\mu$, $\lambda$ and $h^*$.
Note that the condition $\vec{u} \in  \mathcal{G}_{\delta,s}^c$ is equivalent to
\begin{equation}\label{cond_vec_u}
\left|\frac{\beta_t}{m}\sum_{i=1}^m u_i\,+\frac{\gamma(1-\beta_t)}{1 -\gamma} - \nu_M\right|\leq \delta,
\end{equation}
where we used~\eqref{Elength2}.
So $\sum_{i=1}^m u_i$ is of the order
of $m M$. Let  $\delta=\epsilon/2=\beta_{h^*}/4$. Plugging the two inequalities
above back into~\eqref{cond2max}, we see that there is a strictly positive
constant $C'''$, again depending only on $\mu$, $\lambda$ and $h^*$, such that
\begin{align}
\eqref{cond}
&\leq \exp\left(-C''' m \right).
\label{First term}
\end{align}

\medskip

\begin{proof}[Proof of Proposition~\ref{P:Length}]
	Take $\delta=\epsilon/2=\beta_{h^*}/4$. Then apply  \eqref{First term} to the first term on the RHS of \eqref{E:Length2terms} and Lemma \ref{L:E_delta,s} to the second term on the RHS of \eqref{E:Length2terms}.
	The proof of Proposition~\ref{P:Length} is complete by collecting  inequalities \eqref{E:Length2} and \eqref{E:Length2terms}.
\end{proof}

\subsection{Reconstructing the root sequence given the sequence length}

In this subsection, we bound the second term on the RHS of \eqref{Pf_Fk}, which is the probability of incorrectly reconstructing the root sequence, given its length.
Suppose we know that the ancestral sequence length is $M\geq 1$, that is, the ancestral state is an element of $\{A,T,C,G\}^M$. Let $F^{(M)}_{k,s}$ denote the estimator in~\eqref{Def:F^M_k} when $\widehat{M}^{\,k} = M$.
Formally, we prove the following.
\begin{prop}[Sequence reconstruction error]\label{P:Seq}
	For all $M\in \ZZ_+$, $\vec{x}\in \{A,T,C,G\}^M$ with $|\vec{x}| = M$ and $s\in (0,\,h^*)$, 
	\begin{equation*}\label{E:Seq}
	\P^{\,\vec{x}}\left[
	F^{(M)}_{k,s}\left(\vec{X}_{\partial T^k}\right) \neq \vec{x} 
	\right] \leq   8M\exp{\left(\,\frac{-\, |\partial T^{k}(s)|}{8\,\hat{C}}\,\right)} \,+\,C\,M^2\,\hat{C}\,\,s	
	\end{equation*}
	where $\hat{C}:=\|V^{-1}\,\Psi^{-1}\|_{\infty\to \infty}^2\in (0,\infty)$ depends on  $\mu$, $\lambda$, $h^*$ and $\{t_j\}_{j=1}^M$, and $C\in(0,\infty)$ is a constant which depends only on $\nu$, $\mu$ and $\lambda$.
\end{prop}
\begin{remark}
Note that $\hat{C}$ depends implicitly
on $M$.
\end{remark}

\paragraph{Outline of the proof of Proposition \ref{P:Seq}} 
Fix $k\geq 1$, $s\in (0,h^*)$, $M\geq 1$ and $\{t_j\}_{j=1}^M$. 
By the construction of $F^{(M)}_{k,s}$ in \eqref{Def:F^M_k}, our sequence  estimator $F^{(M)}_{k,s}\big(\vec{X}_{\partial T^k}\big)\in \{A,T,C,G\}^M$ is correct if it is close to the argmax over the columns of the matrix $V^{-1}\Psi^{-1} U$. By the definitions of $U,\Psi,V$, our analysis boils down to bounding the deviations of the empirical frequencies $f^{k,s}_{\sigma}(t_j)$ from their expectations $p^{\sigma}_{\vec{x}}(h^*+t_j)$. This is established using similar arguments to that of Proposition \ref{P:Length}. Here is an outline.

\paragraph{(i) From deviations of frequencies to deviations of  matrices } 
Recall that $ Y^{\vec{x}} $ is the $M\times 4$ matrix whose entries are  $1_{\{x_j=\sigma\}}$ 
and recall from~\eqref{Def:F^M_k} that
$$
\left[F_{k,s}\left(\vec{X}_{\partial T^k}\right)\right]_i \in \arg\max\left\{ (V^{-1}\Psi^{-1} U)_{i,\sigma} : \sigma \in \{A,T,C,G\}\right\},
$$
a formula which is motivated by the proof
of Lemma~\ref{T:Initial}.
Hence,
$F^{(M)}_{k,s}\big(\vec{X}_{\partial T^k}\big) = \vec{x}$ is implied by
$\| V^{-1}\Psi^{-1} U- Y^{\vec{x}}\|_{\max}<1/2$, where $\|A\|_{\max}:=\max_{i,j}|a_{ij}|$ is the maximum among the absolute values of all entries of $A$. 
Because we assume that $M$ is known
for the purposes of this analysis,
the matrices $V$ and $\Psi$ are known,
while $U$ depends on the data 
(consult Section~\ref{S:results} for the full definitions).
To turn the inequality above into a condition on $U$, we note that
\begin{align}
\|V^{-1}\Psi^{-1} U-  Y^{\vec{x}}\|_{\max}&=\|V^{-1}\Psi^{-1} (U-\Psi \,V \,Y^{\vec{x}} )\|_{\max}\notag\\
&\leq \|V^{-1}\,\Psi^{-1}\|_{\infty\to \infty}\,\|U-\Psi \,V \, Y^{\vec{x}}\|_{\max},\label{Ineq_norm}
\end{align}
where  $\|A\|_{\infty\to \infty}:=\max_{i}\sum_{j}|a_{ij}|$ is the maximum absolute row sum of a matrix $A=(a_{ij})$. The above two facts give
\begin{align*}
\P^{\,\vec{x}}\left[
F^{(M)}_{k,s}\left(\vec{X}_{\partial T^k}\right) \neq \vec{x} 
\right] 
\leq&\,\P^{\,\vec{x}}\left[
\| V^{-1}\Psi^{-1} U- Y^{\vec{x}} \|_{\max}\geq 1/2\right] \notag\\
\leq&\, \P^{\,\vec{x}}\left[
\|U-\Psi \,V \, Y^{\vec{x}}\|_{\max}\geq \frac{1}{2\,\|V^{-1}\,\Psi^{-1}\|_{\infty\to \infty}}\right] \notag.
\end{align*}
By definitions of $U,\Psi,V$, we have
$$
(\Psi \,V)_{ji} = \psi_j \eta_j^{i-1},
$$
so that
$$
(\Psi \,V\,Y^{\vec{x}})_{j} = \psi_j \sum_{i=1} 1_{\{x_i=\sigma\}} \eta_j^{i-1},
$$
and, by~\eqref{E:Solve for x},
\begin{align*}
p^{\sigma}_{\vec{x}}(h^*+t_j)
&= 
\pi_{\sigma}\,\phi_j\,\big[1-\eta_j^{M}\big]\,+\,\psi_j\sum_{i=1}^{M}1_{\{x_i=\sigma\}}\,\eta_j^{i-1} \,+\,\pi_{\sigma}\gamma\,\eta_j\\
&= 
\pi_{\sigma}\,\phi_j\,\big[1-\eta_j^{M}\big]
\,+(\Psi \,V\,Y^{\vec{x}})_{j} \,+\,\pi_{\sigma}\gamma\,\eta_j.
\end{align*}
Hence, from~\eqref{S:def:U},

$$\|U-\Psi \,V \, Y^{\vec{x}}\|_{\max}=\max_{\sigma\in\{A,T,C,G\}}\max_{1\leq j\leq M} \big|f^{k,s}_{\sigma}(t_j)-p^{\sigma}_{\vec{x}}(h^*+t_j)\big|.$$ 
Therefore we can bound the error probability in terms of the deviations of the empirical frequencies $f^{k,s}_{\sigma}(t_j)$ as follows
\begin{align}\label{Ineq:Seq}
&\P^{\,\vec{x}}\left[
F^{(M)}_{k,s}\left(\vec{X}_{\partial T^k}\right) \neq \vec{x} 
\right] \notag\\
&\qquad \leq \sum_{\sigma\in\{A,T,C,G\}}\sum_{j=1}^M\P^{\,\vec{x}}\left[ \big|f^{k,s}_{\sigma}(t_j)-p^{\sigma}_{\vec{x}}(h^*+t_j)\big|
\geq \frac{1}{2\,\|V^{-1}\,\Psi^{-1}\|_{\infty\to \infty}}\right]. 
\end{align}

\paragraph{(ii) Estimating deviation of frequencies} 
We then bound each term on the RHS of \eqref{Ineq:Seq} using similar arguments to that of Proposition \ref{P:Length}. 
As before we let  $m:=|\partial T^{k}(s)|$ for convenience. Recall from \eqref{Def:freq} that for $t\in[0,\infty)$
\begin{equation*}
f^{k,s}_{\sigma}(t) = \frac{1}{m}\,\sum_{v\in \partial \tilde{T}^{k,s}} p^{\sigma}_{\vec{X}_{v}}(t).
\end{equation*}
All the terms $p^{\sigma}_{\vec{X}_{v}}(t)$, for $v\in \partial \tilde{T}^{k,s}$, have expectation $p^{\sigma}_{\vec{x}}(h^*+t)$ under $\P^{\vec{x}}$, by the Markov property. 
%
%
Therefore, for $\epsilon \in (0,\infty)$ and $t\in[0,\infty)$, we have
\begin{equation}\label{eta_v_con}
\P^{\vec{x}}\left(|f^{k,s}_{\sigma}(t) - p_{\vec{x}}^{\sigma}(h^*+t)|>\epsilon\right) 
= \P^{\vec{x}}\left(\left| \frac{1}{m}\,\sum_{v\in \partial \tilde{T}^{k,s}}y_v\right|>\epsilon\right)
\end{equation}
where
$$y_v:=p^{\sigma}_{\vec{X}_{v}}(t) - p^{\sigma}_{\vec{x}}(h^*+t)\quad \text{for }v\in \partial \tilde{T}^{k,s},$$
are centered but {\it correlated} random variables under $\P^{\vec{x}}$. To bound \eqref{eta_v_con} we use the same method that we used to bound \eqref{E:Length2}, namely by considering the 
conditional expectation of $f^{k,s}_{\sigma}(t)$ given the states $\vec{X}_{\partial T^k(s)}=\{\vec{X}_{v}:\,v\in \partial T^k(s)\}$, which by the Markov property is equal to
\begin{equation}\label{CondExpectation_f}
\E^{\vec{x}}\left[f^{k,s}_{\sigma}(t)\middle|\vec{X}_{\partial T^k(s)}\,\right] = \frac{1}{m}\sum_{u\in \partial T^k(s)}p_{\vec{X}_{u}}^{\sigma}(h^*+t-\,s).
\end{equation}

As before we   first bound the probability of the event that this conditional expectation 
is close to its expectation (which is also $p_{\vec{x}}^{\sigma}(h^*+t)$), then conditioned on this event we establish a concentration inequality for $f^{k,s}_{\sigma}(t)$, based on conditional independence. Since all $y_v$
are  bounded between $1$ and $-1$, we apply Hoeffding's inequality~\cite{Hoeffding:63} for this purpose.


We detail the above argument in Steps A-C as follows.

\paragraph{(A) Decomposition by conditioning on level $s$. } 
Similarly to \eqref{E:Length2terms} we have
\begin{equation}\label{E:Seq2terms}
\P^{\vec{x}}\left(|f^{k,s}_{\sigma}(t) - p_{\vec{x}}^{\sigma}(h^*+t)|>\epsilon\right) 
\leq \,\P^{\vec{x}}\left(\left\{|f^{k,s}_{\sigma}(t) - p_{\vec{x}}^{\sigma}(h^*+t)| > \epsilon\right\}\;\cap\; (\mathcal{F}^t_{\delta,s})^c\right) + \P^{\vec{x}}( \mathcal{F}^t_{\delta,s}),
\end{equation}
where 
\begin{equation}\label{F_delta}
\mathcal{F}^{\,t}_{\delta,s}:= \Big\{\Big|\E^{\vec{x}}\left[f^{k,s}_{\sigma}(t)\middle|\vec{X}_{\partial T^k(s)}\,\right]  -  \,p_{\vec{x}}^{\sigma}(h^*+t)\Big|> \delta \Big\}.
\end{equation}
As we detail next, we then control the two terms on the RHS of \eqref{E:Seq2terms}. The proof of Proposition \ref{P:Seq} will be completed by taking
\begin{equation}\label{Seq_eps}
\epsilon=\frac{1}{2\,\|V^{-1}\,\Psi^{-1}\|_{\infty\to \infty}},\quad \delta:=\epsilon/2  \quad \text{and}\quad t=t_j.
\end{equation}
in \eqref{E:Seq2terms} and combining with \eqref{Ineq:Seq}.

\paragraph{(B) Bounding $\P_{M}( \mathcal{F}^t_{\delta,s})$. }The second term on the RHS of \eqref{E:Seq2terms} is treated in Lemma \ref{L:F_delta,s} below.

\begin{lemma}[Conditional expectation given level $s$]\label{L:F_delta,s}
	There exists a constant $C\in(0,\infty)$ which depends only on $\lambda$, $\nu$ and $\mu$ such that
	for all $\vec{x}\in \{A,T,C,G\}^M$, $\delta\in (0,\infty)$ and $s\in (0,\,h^*)$, we have
	\begin{equation*}
	\sup_{t\in(0,\infty)}\P^{\vec{x}}( \mathcal{F}^{\,t}_{\delta,s}) \,\leq\, \delta^{-2}CMs,
	\end{equation*}
	where event $\mathcal{F}^{\,t}_{\delta,s}$ is defined in \eqref{F_delta}.
\end{lemma}
\begin{proof}
	Similarly to the proof of Lemma \ref{L:E_delta,s} we use Chebyshev's inequality to control the deviation of $m^{-1}\sum_{u\in \partial T^k(s)}p_{\vec{X}_{u}}^{\sigma}(h^*+t-\,s)$. 
	Using~\eqref{S:CS-AM-GM},
	\begin{align}
	\var_{\vec{x}}\left[\sum_{u\in \partial T^k(s)}p_{\vec{X}_{u}}^{\sigma}(h^*+t-\,s)\right] 
	&\leq m^2 \,\var_{\vec{x}}\left[p_{\vec{X}_{u}}^{\sigma}(h^*+t-\,s)\right], \label{VarP_-1}
	\end{align}
	for any $u \in \partial T^k(s)$, where where we used the fact that the $\vec{X}_u$'s for all
	such $u$'s are identically distributed.
	Using the explicit formula \eqref{E:Solve for x} together with~\eqref{S:CS-AM-GM} again,
for each  $u\in\partial T^k(s)$ we have further
	\begin{align}
	&\var_{\vec{x}}\left[p_{\vec{X}_{u}}^{\sigma}(h^*+t-\,s)\right] \notag\\
	&=\var_{\vec{x}}\left[ -\pi_{\sigma}\,\phi(h^*+t-\,s)\,A^{|\vec{X}_{u}|}  +\,\psi(h^*+t-\,s)\sum_{i=1}^{|\vec{X}_{u}|}1_{\{X_{u,i}=\sigma\}}\,A^{i-1} \right] \notag\\
	&\leq 2\,\pi^2_{\sigma}\,\phi^2(h^*+t-\,s)\,\var_{\vec{x}}\Big[A^{|\vec{X}_{u}|}\Big] \notag\\
	&\quad +\,2\,\psi^2(h^*+t-\,s)\,\var_{\vec{x}}\left[\sum_{i=1}^{|\vec{X}_{u}|}1_{\{X_{u,i}=\sigma\}}\,A^{i-1}\right], \label{VarP0}
	\end{align}
	where $A:= \eta(h^*+t-\,s)$. (Note that the $A$ here is distinct from that in Section~\ref{sec:tkf-basic}.)
	\medskip
	
	The term $\var_{\vec{x}}\Big[A^{|\vec{X}_{u}|}\Big]$ can be 
	bounded as follows. Let $E_u$ be the event that the sequence never left state $\vec{x}$ along the unique path from the root to $u$. Then
	$$\P^{\vec{x}}(E_u)=e^{-q_M s},\quad \text{where}\quad q_M:=M\nu+(M+1)\lambda+M\mu,$$
	and
	\begin{equation*}
	   	\E^{\vec{x}}\Big[A^{|\vec{X}_{u}|}\Big]= A^M\,\P^{\vec{x}}(E_u)+ \E^{\vec{x}}\Big[A^{|\vec{X}_{u}|};E_u^c\Big].
	\end{equation*}
	Since $A \in [0,1]$, we have $\E^{\vec{x}}\Big[A^{|\vec{X}_{u}|};E_u^c\Big]\leq \P^{\vec{x}}(E_u^c)=1-e^{-q_M s}$ and therefore
	$$0\leq \E^{\vec{x}}\Big[A^{|\vec{X}_{u}|}\Big]-A^Me^{-q_M s}\leq 1-e^{-q_M s}.$$
	Similarly,
	$$0\leq \E^{\vec{x}}\Big[A^{2|\vec{X}_{u}|}\Big]- A^{2M}e^{-q_M s}\leq 1-e^{-q_M s}.$$
	The last two displays give
	\begin{align}\label{VarP1}
	\var_{\vec{x}}\Big[A^{|\vec{X}_{u}|}\Big]
	&=\E^{\vec{x}}\Big[A^{2|\vec{X}_{u}|}\Big]-\left(\E^{\vec{x}}\Big[A^{|\vec{X}_{u}|}\Big]\right)^2 \notag\\
	&\leq \left(A^{2M}e^{-q_M s}+1-e^{-q_M s}\right) \;-\; (A^M\,e^{-q_M s})^2 \notag\\
	&= A^{2M}e^{-q_M\,s}(1-e^{-q_M s})  +1-e^{-q_M s} \notag\\
	&\leq 2(1-e^{-q_M s}).
	\end{align}
	
	For the second variance term in \eqref{VarP0}, a similar argument gives
	\begin{align*}
	&\E^{\vec{x}}\left[\sum_{i=1}^{|\vec{X}_{u}|}1_{\{X_{u,i}=\sigma\}}\,A^{i-1}\right]\\ 
	&\qquad = \left(\sum_{i=1}^{M}1_{\{x_i=\sigma\}}\,A^{i-1}\right)\,\P^{\vec{x}}(E_u) + \E^{\vec{x}}\left[\sum_{i=1}^{|\vec{X}_{u}|}1_{\{X_{u,i}=\sigma\}}\,A^{i-1};\,E^c_u\cap \{|\vec{X}_{u}|\geq 1\}\right]
	\end{align*}

	and 	
	\begin{align*}
	   0\leq  \E^{\vec{x}}\left[\sum_{i=1}^{|\vec{X}_{u}|}1_{\{X_{u,i}=\sigma\}}\,A^{i-1};\,E^c_u\cap \{|\vec{X}_{u}|\geq 1\}\right]\leq \sum_{i=1}^{\infty}A^{i-1}\P^{\vec{x}}\left(E^c_u\right)=\frac{1-e^{-q_M s}}{1-A}
	\end{align*}
	because  $A \in (0,1)$ when $h^*+t-s>0$. The last two displays immediately give  
	\begin{equation*}
	0\leq \E^{\vec{x}}\left[\sum_{i=1}^{|\vec{X}_{u}|}1_{\{X_{u,i}=\sigma\}}\,A^{i-1}\right] - \left(\sum_{i=1}^{M}1_{\{x_i=\sigma\}}\,A^{i-1}\right)e^{-q_M s} \leq \frac{1-e^{-q_M s}}{1-A}.
	\end{equation*}
	Similarly we have 
	\begin{align*}
	0\leq & \E^{\vec{x}}\left[\left(\sum_{i=1}^{|\vec{X}_{u}|}1_{\{X_{u,i}=\sigma\}}\,A^{i-1}\right)^2\,\right] - \left(\sum_{i=1}^{M}1_{\{x_i=\sigma\}}\,A^{i-1} \right)^2e^{-q_M s}\\
	=& \E^{\vec{x}}\left[\left(\sum_{i=1}^{|\vec{X}_{u}|}1_{\{X_{u,i}=\sigma\}}\,A^{i-1}\right)^2;\,E^c_u\cap \{|\vec{X}_{u}|\geq 1\}\right]\\
	\leq & \left(\sum_{i=1}^{\infty}A^{i-1} \right)^2\,\P^{\vec{x}}\left(E^c_u\right)\\ 
	=& \frac{1-e^{-q_M s}}{(1-A)^2}.
	\end{align*}
	From the above two displays,
	we obtain as in \eqref{VarP1} the following estimate
	\begin{align}\label{VarP2}
	\var_{\vec{x}}\left[\sum_{i=1}^{|\vec{X}_{u}|}1_{\{X_{u,i}=\sigma\}}\,A^{i-1}\right]
	\leq&\, \frac{2(1-e^{-q_M s})}{(1-A)^2}.
	\end{align}
	
	From \eqref{VarP0}-\eqref{VarP2} we have
	\begin{align*}
	&\var_{\vec{x}}\left[p_{\vec{X}_{u}}^{\sigma}(h^*+t-s)\right] \notag\\
	&\leq 4\,\pi^2_{\sigma}\,\phi^2(h^*+t-s)\,(1-e^{-q_M s})\,+\,4\,\psi^2(h^*+t-s)\,\frac{1-e^{-q_M s}}{(1-A)^2}\\
	&= 4\,(1-e^{-q_M s})\,\left[\pi^2_{\sigma}\,\phi^2(h^*+t-s)+\frac{\psi^2(h^*+t-s)}{\left(1-\eta(h^*+t-s)\right)^2}\right]\\
	&\leq 
	4\,(1-e^{-q_M s})\,\left[\|\phi\|_{\infty}^2\pi^2_{\sigma}+ \frac{\gamma^2}{(1-\gamma)^2}\right]\\
	&\leq C_1Ms
	\end{align*}
	for all $s\in (0,h^*)$ and $t\in(0,\infty)$, where $C_1\in(0,\infty)$ is a constant which depends only on $\lambda$, $\nu$ and $\mu$. In the second to last inequality, we used the following facts that follow directly from the definitions: (i) $\eta$ is
	a strictly increasing function with $\eta(0)=0$ and $\lim_{t\to\infty}\eta(t)=1/\gamma$, (ii) the supremum norm $\|\phi\|_{\infty}<\infty$ and (iii) $\psi(t)$ is a decreasing function bounded above by $1$.
	
	The result follows by \eqref{VarP_-1} and Chebyshev's inequality.
\end{proof}

\paragraph{(C) Deviation of $f^{k,s}_{\sigma}(t)$ conditioned on $\mathcal{F}_{\delta,s}^t$}
By the Markov property and conditional independence as in \eqref{cond2max}, the first term on the RHS of~\eqref{E:Seq2terms} is equal to
\begin{align}
&\P^{\vec{x}}\left(\left\{|f^{k,s}_{\sigma}(t) - p_{\vec{x}}^{\sigma}(h^*+t)| > \epsilon\right\}\;\cap\; (\mathcal{F}^t_{\delta,s})^c\right)\notag\\
&\qquad =
\sum_{(\vec{y}_i)_{i=1}^m \in \mathcal{H}_{\delta,s}^c} \, \P^{\vec{x}}\left(\left|\sum_{i=1}^m p^{\sigma}_{\vec{Z}^{(\vec{y}_i)}}(t) - m\,p_{\vec{x}}^{\sigma}(h^*+t)\right| > m\epsilon\right)\P^{\vec{x}}\left(
\vec{X}_{\partial T^k(s)} = (\vec{y}_i)_{i=1}^m
\right)\notag\\
& \qquad \leq \max_{(\vec{y}_i)_{i=1}^m \in \mathcal{H}_{\delta,s}^c} \, \P\left(\left|\sum_{i=1}^m p^{\sigma}_{\vec{Z}^{(\vec{y}_i)}}(t)\;-\;m\,p_{\vec{x}}^{\sigma}(h^*+t)\right| \geq m\epsilon \right). \label{cond2Seq}
\end{align}
where 
\begin{equation}
\mathcal{H}_{\delta,s}^c:=\left\{(\vec{y}_i)_{i=1}^m \in \mathcal{S}^m:\;\left|\sum_{i=1}^mp_{\vec{y}_{i}}^{\sigma}(h^*+t-s)   - m \,p_{\vec{x}}^{\sigma}(h^*+t)\right| \leq  m\delta \right\}\label{S:def:H}
\end{equation}
and $\{\vec{Z}^{(\vec{y}_i)}\}_{i=1}^m$ are {\it independent} copies of the TKF91 process starting from $\{\vec{y}_i\}_{i=1}^m$ respectively, evaluated at time $h^*-s$. By the triangle inequality applied to the
events in~\eqref{cond2Seq}
and~\eqref{S:def:H}, and then Hoeffding's inequality~\cite{Hoeffding:63}, we have
\begin{align}
\eqref{cond2Seq}
&\leq \max_{(\vec{y}_i)_{i=1}^m \in \mathcal{H}_{\delta,s}^c} \, \P\left(\left|\sum_{i=1}^m\left( p^{\sigma}_{\vec{Z}^{(\vec{y}_i)}}(t)\,-\,p_{\vec{y}_{i}}^{\sigma}(h^*+t-s)\right)\right| \geq m(\epsilon-\delta) \right) \notag \\
&\leq 2\,\exp{\left(-2m \,(\epsilon-\delta)^2\right)}. \label{cond2Seq2}
\end{align}

\medskip

\begin{proof}[Proof of Proposition~\ref{P:Seq}]	
	As pointed out in \eqref{Seq_eps}, we will take
	\begin{equation*}
	\epsilon=\frac{1}{2\,\|V^{-1}\,\Psi^{-1}\|_{\infty\to \infty}},\quad \delta:=\epsilon/2  \quad \text{and}\quad t=t_j.
	\end{equation*}
	From \eqref{cond2Seq2} we have
	\begin{align*}
	\eqref{cond2Seq}
	&\leq 2\,\exp{\left( \,\frac{-m}{8\,\|V^{-1}\,\Psi^{-1}\|_{\infty\to \infty}^2}\,\right)}.
	\end{align*}
	Taking $t\in \{t_j\}$ in \eqref{E:Seq2terms}, we see that the terms on the RHS of \eqref{Ineq:Seq} are of the form 
	\begin{align}\label{Ineq:Seq2}
	&\P^{\,\vec{x}}\left[ \big|f^{k,s}_{\sigma}(t_j)-p^{\sigma}_{\vec{x}}(h^*+t_j)\big|
	\geq \frac{1}{2\,\|V^{-1}\,\Psi^{-1}\|_{\infty\to \infty}}\right] \notag\\
	&\qquad\qquad\qquad \leq  2\,\exp{\left( \,\frac{-m}{8\,\|V^{-1}\,\Psi^{-1}\|_{\infty\to \infty}^2}\,\right)}+ \,\P^{\vec{x}}\left[ \mathcal{F}^{\,t_j}_{\delta,s}\right],\notag\\
	&\qquad\qquad\qquad \leq  2\,\exp{\left( \,\frac{-m}{8\,\|V^{-1}\,\Psi^{-1}\|_{\infty\to \infty}^2}\,\right)}+ \,\left(\frac{1}{4\,\|V^{-1}\,\Psi^{-1}\|_{\infty\to \infty}}\right)^{-2}CMs,
	\end{align}
	where the last line comes from
	Lemma \ref{L:F_delta,s} and our choice
	of $\delta$.
	The proof of Proposition \ref{P:Seq} is completed upon plugging into \eqref{Ineq:Seq} and summing over
	$\sigma$ and $j$.
\end{proof}

\subsection{Proof of Theorem \ref{thm:error}}

Applying Propositions~\ref{P:Length} and~\ref{P:Seq} to the first and second terms on the RHS of \eqref{Pf_Fk} respectively, we obtain
constants $C_1,\,C_2,\,C_3\in (0,\infty)$ which depend only on  $h^*$,  $\mu$ and $\lambda$  
such that
\begin{align}\label{TechBound}
\P^{\vec{x}}\left[F_{k,s}(\vec{X}_{\partial T^{k}})\neq \vec{x}\right] 
&\leq   2\exp{\left(-C_1\,|\partial T^{k}(s)|\right)}\,+\,C_2\,M\,s \notag\\
&+  8M\exp{\left(\,\frac{-\, |\partial T^{k}(s)|}{8\,\hat{C}}\,\right)} \,+\,C_3\,M^2\,\hat{C}\,s
\end{align}
for all $s\in (0,h^*/2]$ and $k\in \mathbb{N}$,  
where $M=|\vec{x}|$ and $\hat{C}:=\|V^{-1}\,\Psi^{-1}\|_{\infty\to \infty}^2\in (0,\infty)$ is a constant which depends only on  $\mu$, $\lambda$, $h^*$ and $\{t_j\}_{j=1}^M$. 

For any $\epsilon\in(0,\infty)$, there exists $M_{\epsilon}$ such that $\sum_{\{\vec{x}:\,|\vec{x}|>M_{\epsilon}\}}\Pi(\vec{x}) <\epsilon/2$.
Hence 
\begin{align*}
\P^{\Pi}\left[F_{k,s}(\vec{X}_{\partial T^{k}})\neq \vec{X}_{\rho}\right] &< 
\sum_{\{\vec{x}:\,|\vec{x}|\leq M_{\epsilon}\}}\P^{\vec{x}}\left[F_{k,s}(\vec{X}_{\partial T^{k}})\neq \vec{x}\right]\,\Pi(\vec{x})\,+\frac{\epsilon}{2}\\
&\leq \max_{\{\vec{x}:\,|\vec{x}|\leq M_{\epsilon}\}}\P^{\vec{x}}\left[F_{k,s}(\vec{X}_{\partial T^{k}})\neq \vec{x}\right]\,+\frac{\epsilon}{2}.
\end{align*}
From \eqref{TechBound}, there exist constants $C_4,\,C_5,\, C_6\in (0,\infty)$ which depend only on  $\mu$, $\lambda$, $h^*$, $\epsilon$ and $\{t_j\}_{j=1}^{M_{\epsilon}}$ such that
\begin{align*}
\max_{\{\vec{x}:\,|\vec{x}|\leq M_{\epsilon}\}}\P^{\vec{x}}\left[F_{k,s}(\vec{X}_{\partial T^{k}})\neq \vec{x}\right] \leq
C_4 \exp{\left(\,- C_5\, |\partial T^{k}(s)|\,\right)} \,+\,C_6\,s
\end{align*}
for all $s\in (0,h^*/2]$ and $k\in \mathbb{N}$. Inequality \eqref{error2} follows from this. The proof of Theorem \ref{thm:error} is complete.

%

\subsection{Proof of Theorem \ref{thm:consistent}}

Theorem \ref{thm:consistent} follows from Theorem~\ref{thm:error}
upon taking sequences $\epsilon_m \downarrow 0$ and $s_m \downarrow 0$ and then a subsequence $k_m \to +\infty$ such that
the error in~\eqref{error2} goes to $0$. This is possible thanks to
the big bang condition, which guarantees 
$|\partial T^k(s_m)| \to +\infty$ as $k \to \infty$.

\section{Discussion}
\label{S:discussion}

In this paper, we considered the ancestral sequence reconstruction (ASR) problem in the taxon-rich context for the TKF91 process. 
It has been known from previous work ~\cite[Theorem 1]{MR3814241} that the \textit{Big Bang} condition is necessary for the existence of consistent estimators. In this paper, we design the first estimator which is not only {\it consistent} but also {\it explicit and computationally tractable}. Our ancestral reconstruction algorithm involves two steps: we first estimate the length of the ancestral sequence, then estimate the nucleotides conditioned on the sequence length. 

The novel observation that leads to the design of our estimator is a new constructive proof of {\it initial-state identifiability}, formulated in Lemma \ref{T:Initial}, which says that one can  explicitly invert the mapping from the root sequence to the distribution of the leaf sequences.  This is nontrivial for evolutionary models with indels.
Our estimator is computationally efficient in the sense that the number
of arithmetic operations required scales like a polynomial in the size of the input data. Indeed the length estimator is linear in the number of input sequences and the matrix manipulations in the sequence estimator are polynomial in the length of the longest input sequence. 

We believe this is a first step towards designing practical estimators
with consistency guarantees, which are
lacking under indel models. We leave 
the non-trivial task of implementation
and validation for future work. On the
theoretical side, it would be of interest
to explore whether our techniques 
can be applied to more general indel
models, for instance those allowing
multiple simultaneous insertions 
and deletions.





%

\newpage


\bibliographystyle{alpha}
\bibliography{my,thesis2}

\newpage

\appendix

\section{Some properties of the TKF91 length process}
\label{sec:tkf-basic}

Recall the TKF91 edge process $\mathcal{I}=(\mathcal{I}_t)_{t\geq 0}$ in Definition \ref{Def_TKF91}, which has parameters  $(\nu,\,\lambda,\,\mu)\in (0,\infty)^3$ with $\lambda< \mu$ and $(\pi_A,\,\pi_T,\,\pi_C,\,\pi_G)\in [0,\infty)^4$  with $\pi_A +\pi_T + \pi_C + \pi_G = 1$.
The sequence length of the TKF91 edge process is a continuous-time linear birth-death-immi\-gration process $(|\mathcal{I}_{t}|)_{t\geq 0}$ with infinitesimal generator $Q_{i,i+1}=\lambda+i\lambda$ (for $i\in \ZZ_+$), $Q_{i,i-1}=i\mu$ (for $i\geq 1$) and $Q_{i,j}=0$ otherwise. This is a well-studied process for which explicit forms for the transition density $p_{ij}(t)$ and probability generating functions $G_i(z,t)=\sum_{j=0}^{\infty}p_{ij}(t)z^j$ are known. See, for instance, \cite[Section 3.2]{Anderson:91} or \cite[Chapter 4]{karlin1981second} for more details. This process was also analyzed in~\cite{thatte2006invertibility} in the related context of phylogeny estimation.

We collect here a few properties that will be useful in our analysis.
The probability generating function is given by
\begin{equation*}
G_i(z,t)=\left[\frac{1-\beta -z(\gamma-\beta)}{1-\beta\gamma -\gamma z(1-\beta)}\right]^i\,\left[\frac{1-\gamma}{1-\beta\gamma -\gamma z(1-\beta)}\right] 
\end{equation*}
for  $i\in \ZZ_+$ and $t>0$, where 
\begin{equation*}
\beta =\beta_t = e^{-(\mu-\lambda)t}\quad \text{and}\quad\gamma =\frac{\lambda}{\mu}.
\end{equation*}
Fix $t\geq 0$ and let $\varphi_i(\theta)=\E_i[e^{\theta\,|\mathcal{I}_{t}|}]$ be the characteristic function of  $|\mathcal{I}_{t}|$ starting at $i$. Then for $\lambda \neq \mu$ (i.e. $\gamma\neq 1$),
\begin{align}
\varphi_i(\theta)&=G_i(e^{\theta},t)=\left[\frac{1-\beta -e^{\theta}(\gamma-\beta)}{1-\beta\gamma -\gamma e^{\theta}(1-\beta)}\right]^i\,\left[\frac{1-\gamma}{1-\beta\gamma -\gamma e^{\theta}(1-\beta)}\right]\notag\\
&=(1-\gamma)\frac{A^i}{B^{i+1}}, \label{varphi}
\end{align}
where 
\begin{equation*}
A=1-\beta -e^{\theta}(\gamma-\beta) \quad\text{and}\quad B=1-\beta\gamma -\gamma e^{\theta}(1-\beta).
\end{equation*}
Differentiating with respect to $\theta$ gives
\begin{align}
	\varphi'_i(\theta)&=(1-\gamma)\,\frac{A^{i-1}\,e^{\theta}\,\big\{i(\beta-\gamma)\,B -(i+1)\gamma(\beta-1)\,A\big\}}{B^{i+2}} \notag \\
	&=(1-\gamma)\,\frac{A^{i-1}\,e^{\theta}\,\big\{\beta(1-\gamma)^2i +\gamma(1-\beta)^2 +  e^{\theta}(\gamma -\beta)\gamma (\beta-1)\big\}}{B^{i+2}}  \notag\\
	&= \frac{A^{i-1}}{B^{i+2}}(Ce^{\theta}+De^{2\theta})\label{varphi'},
\end{align}
where 
\begin{equation*}
C=(1-\gamma)[\beta(1-\gamma)^2i +\gamma(1-\beta)^2] \quad\text{and}\quad D=(1-\gamma)(\gamma -\beta)\gamma (\beta-1).
\end{equation*}
Differentiating with respect to $\theta$ once more gives
\begin{align*}
\varphi''_i(\theta)&=\frac{A^{i-2}e^{\theta}}{B^{i+3}}\Big\{BA(C+2De^{\theta}) \,+\,(Ce^{\theta}+De^{2\theta})\big[(i-1)B(\beta-\gamma)+(i+2)A\gamma(1-\beta)\big]\Big\} 
\end{align*}
The expected value and the second moment are given by 
\begin{align*}
\E_{i}[|\mathcal{I}_{t}|]=\varphi'_i(0)&=
i\beta+\frac{\gamma(1-\beta)}{1 -\gamma}=i\beta+\frac{\lambda}{\lambda -\mu}(\beta-1) \quad\text{ and} 
\\
\E_{i}[|\mathcal{I}_{t}|^2]=\varphi''_i(0)&= i^2\beta^2+i\,\frac{\beta(3\gamma+1)(1-\beta)}{1-\gamma} +\frac{(1-\beta)\gamma[1+(1-2\beta)\gamma]}{(1-\gamma)^2}   
\end{align*}
From these we also have the variance
\begin{align}
\var_i(|\mathcal{I}_{t}|)&=\E_i[|\mathcal{I}_{t}|^2]-(\E_i[|\mathcal{I}_{t}|])^2 \notag\\
&=i^2\beta^2+i\,\frac{\beta(3\gamma+1)(1-\beta)}{1-\gamma} +\frac{(1-\beta)\gamma[1+(1-2\beta)\gamma]}{(1-\gamma)^2}- \left(i\beta+\frac{\gamma(1-\beta)}{1 -\gamma}\right)^2\notag\\
&= i\,\frac{\beta(1-\beta)(1+3\gamma-2\beta)}{1-\gamma}\,+\,\frac{\gamma(1-\beta)(1-\beta\gamma)}{(1-\gamma)^2}.\label{Varlength'}
\end{align}

Consider the function
\begin{align*}
\psi_i(\theta)&=\varphi'_i(0)-\frac{\varphi'_i(\theta)}{\varphi_i(\theta)}\\
&=i\beta+\frac{\gamma(1-\beta)}{1 -\gamma} \,-\, \frac{Ce^{\theta}+De^{2\theta}}{(1-\gamma)AB} \notag\\
&=\left(\beta-\frac{C_1\,e^{\theta}}{(1-\gamma)AB}\right)i\,+\,\frac{\gamma(1-\beta)}{1 -\gamma} \,-\, \frac{C_2\,e^{\theta}+De^{2\theta}}{(1-\gamma)AB} \notag\\
&=\left(\beta-\frac{\beta(1-\gamma)^2\,e^{\theta}}{AB}\right)i\,+\,\frac{\gamma(1-\beta)}{1 -\gamma} \,-\, \frac{\gamma(1-\beta)^2\,e^{\theta}+(\gamma -\beta)\gamma (\beta-1)\,e^{2\theta}}{AB} \notag\\
&= F(\theta)i+G(\theta)
\end{align*} 
where we wrote $C=C_1i+C_2$, with $C_1=\beta(1-\gamma)^3$ and $C_2=(1-\gamma)\gamma(1-\beta)^2$, and the last line is a definition. Functions $F$ and $G$ can be simplified to
\begin{align}
F(\theta)&=
 \frac{-\,(e^{\theta} - 1) \beta(1-\beta)[(1-\beta\gamma)-e^{\theta}\gamma(\gamma-\beta)]}{AB}\label{Ftheta}\\
G(\theta)&=  \frac{-\,(e^{\theta}-1) \gamma(1-\beta)(1-\beta\gamma)[(1-\beta)-e^{\theta}(\gamma-\beta)]}{(1-\gamma)AB}.\label{Gtheta}
\end{align}

Since $\varphi_i(0)=1$, we have $\psi(0)=0$. We consider the case $\mu\in (\lambda,\infty)$, that is, $\gamma\in(0,1)$. Then both $A$ and $B$ are strictly positive for all $t\in [0,\infty)$, provided that $e^\theta<\mu/\lambda$. Moreover, $F$ and $G$ are continuous on $[0,\,\mu/\lambda)$, smooth on $(0,\,\mu/\lambda)$ and $F(0)=G(0)=0$. 

\section{Notation}
\label{S:Notation}
For the reader's convenience, we list some frequently used notation here.

\begin{table}[]
	\begin{tabular}{ll}
	$\mathcal{S}$	&  state space of the TKF 91 edge process, defined in \eqref{S} \\
	$|\vec{x}|$	&  length of the sequence $\vec{x}\in \mathcal{S}$ \\
	$\lambda$ & insertion rate \\
	$\mu$ & deletion rate \\
	$\nu$ & substitution rate \\
	$\Pi$ & stationary distribution  \\
	$\gamma$ & $\lambda/\mu \in (0,1)$ \\	
	$\beta_t$ & $e^{-(\mu-\lambda)t}$\\
	$p^{\sigma}_{\vec{x}}(t)$ &  defined in	\eqref{E:Q^sigma} which, by \eqref{Def:theta_v}, is the probability that the \textit{first nucleotide}\\ 
	 & of a sequence at time $t$ is $\sigma$, under the TKF 91 edge process starting at $\vec{x}$
	\end{tabular}
		\caption{TKF 91 edge process  $\mathcal{I}=(\mathcal{I}_t)_{t\geq 0}$ in Definition \ref{Def_TKF91}}.
\end{table}

\begin{table}[]
\begin{tabular}{ll}
$\partial T$	&  leaf set of tree $T$ \\
$\Gamma_T$	&  set of points of tree $T$ \\
$\ell_e$ &length of edge $e$\\
$\ell_g$ & length of the unique path from the root $\rho$ to a point $g\in \Gamma_T$ \\
$T(s)$ & truncation of tree $T$ defined as $\{g\in \Gamma_T:\,\ell_g\leq s\}$ \\
$\vec{X}_{\rho}$ & root state\\
$\vec{X}_{\partial T}$ & leaf states\\
$\{T^k\}$ & sequence of trees satisfying Assumption \ref{A:T^k}\\
$h^k$ & height of tree $T^k$ \\
$h^*$ & uniform bound $\sup_{k}h^{k}$  defined in \eqref{Uniformh} \\
$\widehat{M}^{\,k}$ &  root sequence-length estimator defined in \eqref{Def:M^k*}\\
\end{tabular}
\caption{TKF 91 process on trees}
\end{table}

\begin{table}[]
\begin{tabular}{ll}
$\P^{\vec{x}}$ & conditional probability measure when the root state is $\vec{x}$ \\
$\P^{\pi}$ & conditional probability measure when the root state has distribution $\pi$ \\
$\P_{M}$ & conditional probability measure when the root state has length $M$ \\
$\E^{\vec{x}}$ & expectation with respect to $\P^{\vec{x}}$ \\
$\E^{\pi}$ & expectation with respect to $\P^{\pi}$ \\
$\E_{M}$ & expectation with respect to $\P_{M}$ \\
$[[x]]$ & the unique integer such that $[[x]]-1/2 \leq x < [[x]]+1/2$\\
$F_{k,s}$ & our root estimator constructed in \eqref{Def:F^M_k}\\
$a\wedge b$ & $\min\{a,\,b\}$
\end{tabular}
\caption{Probabilities, expectations and other notation}
\end{table}

\end{document}